\documentclass{article}

\usepackage{amsmath}
\usepackage{amsthm}
\usepackage[T1]{fontenc}
\usepackage{amsfonts}
\usepackage{amssymb}
\usepackage{float}
\usepackage{url}
\usepackage{tkz-graph}
\usepackage{tkz-berge}
\usepackage{pxfonts}
\usepackage{bbm}

\usepackage{tikz}
\usetikzlibrary{automata}
\usetikzlibrary{shadows}
\usetikzlibrary{arrows}
\usetikzlibrary{shapes}
\usetikzlibrary{decorations.pathmorphing}
\usetikzlibrary{fit}

\usepackage{stmaryrd} 
\newcommand{\Intervalle}[2]{\llbracket #1,#2 \rrbracket}
\tikzstyle{every picture}=[
  >=stealth', 
shorten >=1pt, 
node distance=1.44cm,
bend angle=45,
initial text=,
  every state/.style={inner sep=0.75mm, minimum size=1mm},
font=\scriptsize,
]

\newcommand{\Pp}[2]{\begin{tikzpicture}[baseline=0pt,font=\normalsize]
  \node[anchor=base] (V0) {};
  \draw[decorate,decoration={snake,amplitude=.4mm,segment length=2mm}](0,0.1) -- (0.5,0.1);
\end{tikzpicture}_{#1}\big(#2\big)}

\newcommand{\Ppn}[1]{\begin{tikzpicture}[baseline=0pt,font=\normalsize]
  \node[anchor=base] (V0) {};
  \draw[decorate,decoration={snake,amplitude=.4mm,segment length=2mm}](0,0.1) -- (0.5,0.1);
\end{tikzpicture}_{#1}}

\newcommand{\Card}[1]{\# #1}

\newcommand{\Dec}[2]{{\displaystyle\mathop\ggg^{#1}}(#2)}
\newcommand{\Decn}[0]{\ggg}
\newcommand{\Decnk}[1]{{\displaystyle\mathop\ggg^{#1}}}
\newcommand{\Shift}[2]{{\displaystyle\mathop\circleright^{#1}(#2)}}
\newcommand{\ShiftDiam}[2]{{\displaystyle\mathop\Diamondright^{#1}(#2)}}
\newcommand{\ShiftDiamnk}[1]{{\displaystyle\mathop\Diamondright^{#1}}}
\newcommand{\Shiftnk}[1]{{\displaystyle\mathop\circleright^{#1}}}
\newcommand{\Shiftn}[0]{\circleright}

\newtheorem{definition}{Definition}
\newtheorem{lemma}{Lemma}
\newtheorem{example}{Example}
\newtheorem{proposition}{Proposition}
\newtheorem{theorem}{Theorem}

\newtheorem{corollary}{Corollary}

\def\goth{\mathfrak}
\def\N{{\mathbb N}}

\begin{document} 

\title{Some Combinatorial Operators in Language Theory}
\author{J.-G.~Luque, L.~Mignot and F.~Nicart}
\date{\today}

 \maketitle
 
 \tableofcontents
 
\begin{abstract}
Multitildes are regular operators that were introduced by Caron \emph{et al.} in order to increase the number of Glushkov automata.
In this paper, we study the family of the multitilde operators from an algebraic point of view using the notion of operad. This leads to a combinatorial description of already known results as well as new  results on  compositions, actions and enumerations.
\end{abstract}

\section{Introduction}
Regular expressions have been studied from several years since they have numerous applications: pattern matching, compilation, verification, bio-informatics \emph{etc.}.
From a well known result (Kleene \cite{Kle56}), regular languages can be represented by both automata and regular expressions. From an expression, it is possible to compute an automaton whose number of states is a linear function of the alphabetical width (\emph{i.e.} the number of occurrences of alphabet symbols) \cite{MY60,Glu61,Ant96,CZ01a}. In the opposite direction, there exists no construction providing the linearity. For instance, Ehrenfeucht and Zeiger \cite{EZ76} showed a one parameter family of automata whose shortest equivalent regular expressions have a width  exponentially growing with the numbers of states. However this property occurs when the automaton is the Gluskov automaton of an expression. Note the characterization of such an automaton is due to Caron and Ziadi \cite{CZ97}.\\
  Multitilde operators have been introduced by Caron~\emph{et al.}~\cite{CCM11e} in the aim to increase expressiveness of expressions for a bounded length. The study of the equivalence of two multitilde expressions leads naturally to a notion of composition. On the other hand, there is an algebraic way to represent the compositions of operators, namely operads. The aim of this paper is to describe some properties of multitildes in terms of operads.

Usually, operads are used to encode types of algebras by describing the universal operations which act on the elements of any algebras of a given type together with the ways of composing them. To be more concise, an operad is given by a set of symbols (the operations) and a composition law which satisfies some rules such as associativity. Such a study of the compositions of operations appeared earlier in the work of M. Lazard \cite{Laz55} and was referred as \emph{analyseurs}.
The word \emph{operad} is the contraction of \emph{operations} and \emph{monad}. This terminology first appeared  in the field of algebraic topology in a paper of  May about the geometry of iterated loop spaces \cite{May72}. The notion arose  simultaneously in Kelly's categorical works on coherence and the definitions used in \cite{May72} have been constructed after conversations between the two mathematicians.\\
In the 1990's, the interest in the topic was renewed  with the works of Ginzburg and Kapranov  on the Koszul duality \cite{GK94}. 
One of the recent developments of this theory deals with Hopf algebras and combinatorics \cite{Cha07}.\\	
Readers interested by the history of this notion and its recent development can refer to  \cite{May96,Lei04,LV10}.

Several materials about operads are recalled in Section \ref{operad}.
In his PhD thesis \cite{Mig10}, one of the authors  introduced a composition on multitildes. We show (Section \ref{multitilde})
that this confers an operadic structure to the set of multitildes. Throughout the paper, we describe derived structures (isomophic operads, suboperads and quotient operads) which allow us to model several behaviors and properties of the multitilde operators. For instance, we define in Section \ref{SecActLan} an operadic structure on boolean vectors and use it to describe the action of multitildes on languages. In Section \ref{SecTRAPO}, we introduce an operadic structure on POSets in the aim to formalize the equivalence of the actions of two multitildes.
In Section \ref{SecCoPer}, we show that the representation by POSets is optimal in the sense that two different operators have different actions on $k$-tuples of languages. As a consequence, we enumerate the different operations which can be encoded by a multitilde. Finally, we give applications to the representation of finite languages and  also regular languages.

\section{What are operads?\label{operad}}
Operad theory has been developed in the aim to study prototypical algebras that model classical properties such as commutativity, associativity \emph{etc.}. In most cases, operads are considered to study algebras with an underlying vector space but the structure can be more generally used to study $\mathbb K$-modules on monoid ($\mathbb K$ semiring). In the most general context, operads are described in terms of category theory. In our case, since we apply the operad theory to the theory of languages, we will set $\mathbb K=\mathbb B$, the boolean semiring. We have adapted the definitions of this section to the boolean context to give a simplified version of the theory but most of them admits well-known generalizations. More precisely, we describe Set-operads that are operads whose underlying graded spaces are simply sets.
\subsection{Operadic structures}
Operads are algebraic structures which mimic the composition of $n$-ary operators. More explicitly, the construction starts with an underlying graded set ${\goth P}=\bigcup_{n\in\N}{\goth P}(n)$; elements of ${\goth P}(n)$ are called $n$-ary operations. This set is endowed with  functions 
$$\circ: {\goth P}(n)\times{\goth P}(k_1)\times\dots\times{\goth P}(k_n)\rightarrow {\goth P}(k_1+\dots+k_n)$$
called compositions and has a special element ${\bf 1}\in{\goth P}(1)$ called identity. This confers to ${\goth P}$ a structure of multicategory with one object \cite{Lei04}; more precisely the operations of the operad are the morphisms of the multicategory.\\
Furthermore, the compositions satisfy the two rules:
\begin{enumerate}
\item Associativity:
\begin{equation}\label{associativity}
\begin{array}{l}{\bf p}\circ\left({\bf p}_1\circ \left({\bf p}_{1,1},\dots,{\bf p}_{1,k_1}\right),\dots,
{\bf p}_n\circ \left({\bf p}_{n,1},\dots,{\bf p}_{n,k_n}\right)\right)=\\\left({\bf p}\circ \left({\bf p}_1,\dots,{\bf p_n}\right)\right)\circ\left({\bf p}_{1,1},\dots,{\bf p}_{1,k_1},\dots,{\bf p}_{n,1},\dots,{\bf p}_{n,k_n}\right)\end{array}
\end{equation}
\item Identity:
\begin{equation}\label{identity}
{\bf p}\circ\left({\bf 1},\dots,{\bf 1}\right)={\bf 1}\circ {\bf p}={\bf p}.
\end{equation}
\end{enumerate}
The structure of operad is easier to manipulate if we introduce partial composition operations $\circ_i$ which split the composition. These operations are defined by:
\begin{equation}\label{circi}
\begin{array}{rrcl}
\circ_i:&{\goth P}(m)\times {\goth P}(n)&\rightarrow&{\goth P}(m+n-1)\\
&({\bf p}_1,{\bf p}_2)&\rightarrow&{\bf p_1}\circ_i{\bf p_2}:={\bf p_1}\circ\left(\underbrace{{\bf 1},\dots,{\bf 1}}_{i-1\times},{\bf p}_2,\underbrace{{\bf 1},\dots,{\bf 1}}_{m-i-1\times}\right) 
\end{array}\end{equation}
when $1\leq i\leq n$.\\
Let ${\bf p}_1\in{\goth P}(m)$, ${\bf p}_2\in{\goth P}(n)$ and ${\bf p}_3\in{\goth P}(q)$. The partial compositions satisfy the two associative rules which are deduced from the definition (\ref{circi}) and the associativity of the compositions (\ref{associativity}):
\begin{enumerate}
\item Associativity 1:\\
If $1\leq j<i\leq n$  then
\begin{equation}\label{ass1}
\left({\bf p}_1\circ_i {\bf p}_2\right)\circ_j {\bf p}_3=\left({\bf p}_1\circ_j {\bf p}_3\right)\circ_{i+q-1} {\bf p}_2
\end{equation}
\item Associativity 2:\\
If $i\leq n$ then 
\begin{equation}\label{ass2}
{\bf p}_1\circ_j\left({\bf p}_2\circ_i{\bf p}_3\right)=\left({\bf p}_1\circ_{j}{\bf p}_2\right)\circ_{i+j-1}{\bf p}_3.
\end{equation}
\end{enumerate}
The composition is graphically interpreted by grafting trees together; a $n$-array operation is represented by a $n$-ary tree and a composition $\circ_i$ consists in grafting  the root of a tree onto the $i$th leaf of another tree, the resulting tree is associated to an element of ${\goth P}$.\\
Consider two operations ${\bf p}_1\in{\goth P}(m)$ and ${\bf p}_2\in{\goth P}(n)$:

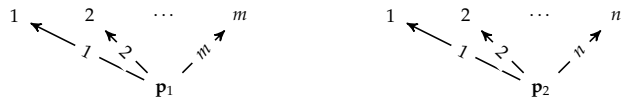
\begin{center}
\begin{tikzpicture}
\tikzstyle{VertexStyle}=[
]

\SetUpEdge[lw = 1.5pt,
 style={post},
labelstyle={sloped}
]
\tikzset{EdgeStyle/.style={post}}

\Vertex[x=0, y=0, 
 L={${\bf p}_1$},style={shape=rectangle}
]{p1}
\Vertex[x=-2, y=1, 
 L={$1$}
]{p11}
\Vertex[x=-1, y=1, 
 L={$2$}
]{p12}

\Vertex[x=0, y=1, 
 L={$\dots$}
]{dd1}

\Vertex[x=1, y=1, 
 L={$m$}
]{p1m}
\Edge[label=$1$,color=black](p1)(p11)
\Edge[label=$2$,color=black](p1)(p12)
\Edge[label=$m$,color=black](p1)(p1m)

\Vertex[x=5, y=0, 
 L={${\bf p}_2$},style={shape=rectangle}
]{p2}
\Vertex[x=3, y=1, 
 L={$1$}
]{p21}
\Vertex[x=4, y=1, 
 L={$2$}
]{p22}

\Vertex[x=5, y=1, 
 L={$\dots$}
]{dd2}

\Vertex[x=6, y=1, 
 L={$n$}
]{p2n}
\Edge[color=black,label=$1$](p2)(p21)
\Edge[color=black,label=$2$](p2)(p22)
\Edge[color=black,label=$n$](p2)(p2n)
\end{tikzpicture}
\end{center}
The composition ${\bf p}_1\circ_i{\bf p}_2$ is represented by the tree:

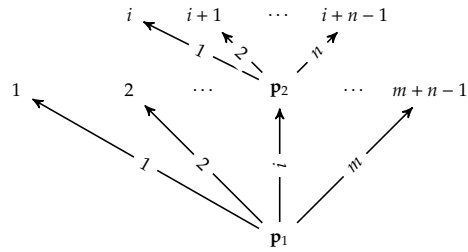
\begin{center}
\begin{tikzpicture}
\tikzstyle{VertexStyle}=[
]

\SetUpEdge[lw = 1.5pt,
 style={post},
labelstyle={sloped}
]
\tikzset{EdgeStyle/.style={post}}

\Vertex[x=0, y=-1, 
 L={${\bf p}_1$},style={shape=rectangle}
]{p1}
\Vertex[x=-3.5, y=1, 
 L={$1$}
]{p11}
\Vertex[x=-2, y=1, 
 L={$2$}
]{p12}

\Vertex[x=-1, y=1, 
 L={$\dots$}
]{dd1}

\Vertex[x=0, y=1, 
 L={${\bf p}_2$}
]{p2}

\Vertex[x=1, y=1, 
 L={$\dots$}
]{dd2}

\Vertex[x=2, y=1, 
 L={$m+n-1$}
]{p1m}

\Vertex[x=-2, y=2, 
 L={$i$}
]{p21}

\Vertex[x=-1, y=2, 
 L={$i+1$}
]{p22}

\Vertex[x=0, y=2, 
 L={$\dots$}
]{dd3}

\Vertex[x=1, y=2, 
 L={$i+n-1$}
]{p2n}

\Edge[color=black,label=$1$](p2)(p21)
\Edge[color=black,label=$2$](p2)(p22)
\Edge[color=black,label=$n$](p2)(p2n)

\Edge[color=black,label=$1$](p1)(p11)
\Edge[color=black,label=$2$](p1)(p12)
\Edge[color=black,label=$m$](p1)(p1m)
\Edge[color=black,label=$i$](p1)(p2)

\end{tikzpicture}
\end{center}

Now the two associativity rules are easily understood in terms of trees:
\begin{enumerate}
\item Associativity 1:
The left hand side of (\ref{ass1}) gives
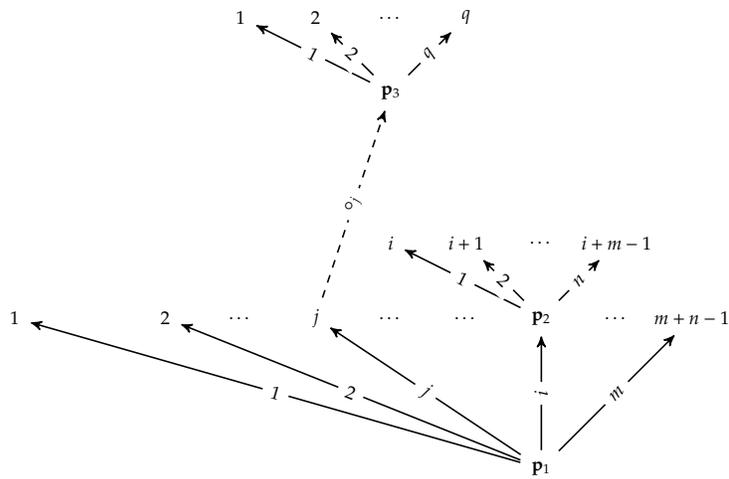
\begin{center}
\begin{tikzpicture}
\tikzstyle{VertexStyle}=[
]

\SetUpEdge[lw = 1.5pt,
 style={post},
labelstyle={sloped}
]
\tikzset{EdgeStyle/.style={post}}

\Vertex[x=0, y=-1, 
 L={${\bf p}_1$},style={shape=rectangle}
]{p1}
\Vertex[x=-7, y=1, 
 L={$1$}
]{p11}
\Vertex[x=-5, y=1, 
 L={$2$}
]{p12}

\Vertex[x=-4, y=1, 
 L={$\dots$}
]{dd6}

\Vertex[x=-3, y=1, 
 L={$j$}
]{p1j}

\Vertex[x=-2, y=1, 
 L={$\dots$}
]{dd7}

\Vertex[x=-1, y=1, 
 L={$\dots$}
]{dd1}

\Vertex[x=0, y=1, 
 L={${\bf p}_2$}
]{p2}

\Vertex[x=1, y=1, 
 L={$\dots$}
]{dd2}

\Vertex[x=2, y=1, 
 L={$m+n-1$}
]{p1m}

\Vertex[x=-2, y=2, 
 L={$i$}
]{p21}

\Vertex[x=-1, y=2, 
 L={$i+1$}
]{p22}

\Vertex[x=0, y=2, 
 L={$\dots$}
]{dd3}

\Vertex[x=1, y=2, 
 L={$i+m-1$}
]{p2n}

\Edge[color=black,label=$1$](p2)(p21)
\Edge[color=black,label=$2$](p2)(p22)
\Edge[color=black,label=$n$](p2)(p2n)

\Edge[color=black,label=$1$](p1)(p11)
\Edge[color=black,label=$2$](p1)(p12)
\Edge[color=black,label=$m$](p1)(p1m)
\Edge[color=black,label=$i$](p1)(p2)
\Edge[color=black,label=$j$](p1)(p1j)

\Vertex[x=-2, y=4, 
 L={${\bf p}_3$},style={shape=rectangle}
]{p3}
\Vertex[x=-4, y=5, 
 L={$1$}
]{p31}
\Vertex[x=-3, y=5, 
 L={$2$}
]{p32}

\Vertex[x=-2, y=5, 
 L={$\dots$}
]{dd3}

\Vertex[x=-1, y=5, 
 L={$q$}
]{p3q}

\Edge[color=black,label=$1$](p3)(p31)
\Edge[color=black,label=$2$](p3)(p32)
\Edge[color=black,label=$q$](p3)(p3q)

\Edge[color=black,style={dashed},label=$\circ_j$](p1j)(p3)

\end{tikzpicture}
\end{center}
whilst the right hand side reads

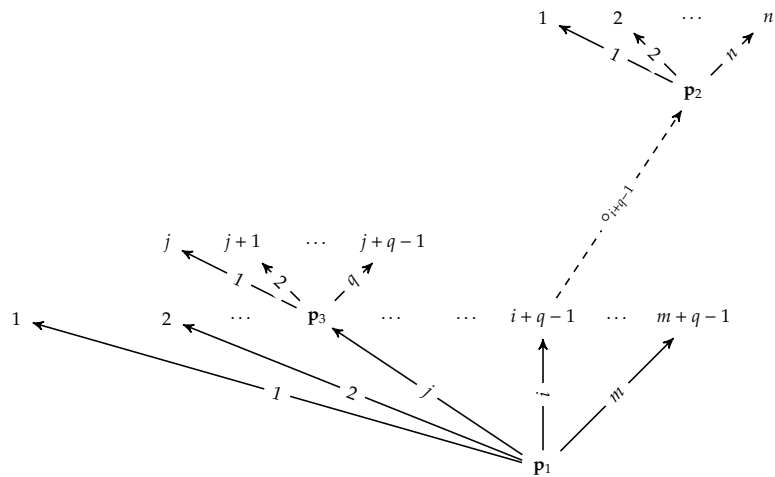
\begin{center}
\begin{tikzpicture}
\tikzstyle{VertexStyle}=[
]

\SetUpEdge[lw = 1.5pt,
 style={post},
labelstyle={sloped}
]
\tikzset{EdgeStyle/.style={post}}

\Vertex[x=0, y=-1, 
 L={${\bf p}_1$},style={shape=rectangle}
]{p1}
\Vertex[x=-7, y=1, 
 L={$1$}
]{p11}
\Vertex[x=-5, y=1, 
 L={$2$}
]{p12}

\Vertex[x=-4, y=1, 
 L={$\dots$}
]{dd6}

\Vertex[x=0, y=1, 
 L={$i+q-1$}
]{p1i}

\Vertex[x=-2, y=1, 
 L={$\dots$}
]{dd7}

\Vertex[x=-1, y=1, 
 L={$\dots$}
]{dd1}

\Vertex[x=-3, y=1, 
 L={${\bf p}_3$}
]{p3}

\Vertex[x=1, y=1, 
 L={$\dots$}
]{dd2}

\Vertex[x=2, y=1, 
 L={$m+q-1$}
]{p1m}

\Vertex[x=-5, y=2, 
 L={$j$}
]{p31}

\Vertex[x=-4, y=2, 
 L={$j+1$}
]{p32}

\Vertex[x=-3, y=2, 
 L={$\dots$}
]{dd3}

\Vertex[x=-2, y=2, 
 L={$j+q-1$}
]{p3q}

\Vertex[x=2, y=4, 
 L={${\bf p}_2$},style={shape=rectangle}
]{p2}
\Vertex[x=0, y=5, 
 L={$1$}
]{p21}
\Vertex[x=1, y=5, 
 L={$2$}
]{p22}

\Vertex[x=2, y=5, 
 L={$\dots$}
]{dd3}

\Vertex[x=3, y=5, 
 L={$n$}
]{p2n}

\Edge[color=black,label=$1$](p3)(p31)
\Edge[color=black,label=$2$](p3)(p32)
\Edge[color=black,label=$q$](p3)(p3q)

\Edge[color=black,label=$1$](p2)(p21)
\Edge[color=black,label=$2$](p2)(p22)
\Edge[color=black,label=$n$](p2)(p2n)

\Edge[color=black,label=$1$](p1)(p11)
\Edge[color=black,label=$2$](p1)(p12)
\Edge[color=black,label=$m$](p1)(p1m)
\Edge[color=black,label=$i$](p1)(p1i)
\Edge[color=black,label=$j$](p1)(p1j)

\Edge[color=black,style={dashed},label=$\circ_{i+q-1}$](p1i)(p2)

\end{tikzpicture}
\end{center}

\item Associativity 2: Drawing the left hand side of (\ref{ass2}), we find
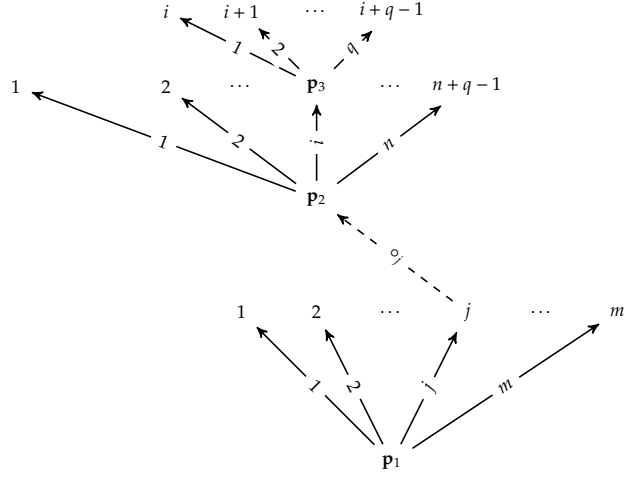
\begin{center}
\begin{tikzpicture}
\tikzstyle{VertexStyle}=[
]

\SetUpEdge[lw = 1.5pt,
 style={post},
labelstyle={sloped}
]
\tikzset{EdgeStyle/.style={post}}

\Vertex[x=-1,y=-1,L={${\bf p}_1$}]{p1}
\Vertex[x=-3,y=1,L={$1$}]{p11}
\Vertex[x=-2,y=1,L={$2$}]{p12}
\Vertex[x=0,y=1,L={$j$}]{p1j}
\Vertex[x=2,y=1,L={$m$}]{p1m}
\Vertex[x=-1,y=1,L={$\dots$}]{dd1}
\Vertex[x=1,y=1,L={$\dots$}]{dd2}

\Vertex[x=-2,y=2.5,L={${\bf p}_2$}]{p2}
\Vertex[x=-6,y=4,L={$1$}]{p21}
\Vertex[x=-4,y=4,L={$2$}]{p22}
\Vertex[x=-2,y=4,L={${\bf p}_3$}]{p3}
\Vertex[x=0,y=4,L={$n+q-1$}]{p2n}
\Vertex[x=-3,y=4,L={$\dots$}]{dd3}
\Vertex[x=-1,y=4,L={$\dots$}]{dd4}

\Vertex[x=-4,y=5,L={$i$}]{p31}
\Vertex[x=-3,y=5,L={$i+1$}]{p32}
\Vertex[x=-2,y=5,L={$\dots$}]{dd5}
\Vertex[x=-1,y=5,L={$i+q-1$}]{p3q}

\Edge[label={$1$}](p1)(p11)
\Edge[label={$2$}](p1)(p12)
\Edge[label={$j$}](p1)(p1j)
\Edge[label={$m$}](p1)(p1m)

\Edge[label={$1$}](p2)(p21)
\Edge[label={$2$}](p2)(p22)
\Edge[label={$i$}](p2)(p3)
\Edge[label={$n$}](p2)(p2n)

\Edge[label={$1$}](p3)(p31)
\Edge[label={$2$}](p3)(p32)
\Edge[label={$q$}](p3)(p3q)

\Edge[color=black,style={dashed},label=$\circ_{j}$](p1j)(p2)
\end{tikzpicture}
\end{center}
and the right hand side of (\ref{ass2}) gives

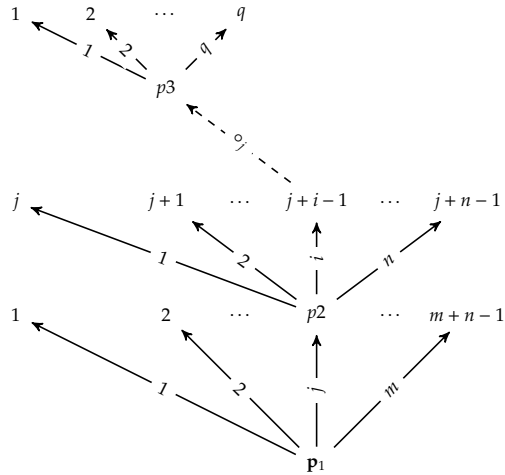
\begin{center}
\begin{tikzpicture}
\tikzstyle{VertexStyle}=[
]

\SetUpEdge[lw = 1.5pt,
 style={post},
labelstyle={sloped}
]
\tikzset{EdgeStyle/.style={post}}

\Vertex[x=0,y=-1,L={${\bf p}_1$}]{p1}
\Vertex[x=-4,y=1,L={$1$}]{p11}
\Vertex[x=-2,y=1,L={$2$}]{p12}
\Vertex[x=0,y=1,L={$p2$}]{p2}
\Vertex[x=2,y=1,L={$m+n-1$}]{p1m}
\Vertex[x=-1,y=1,L={$\dots$}]{dd1}
\Vertex[x=1,y=1,L={$\dots$}]{dd2}

\Vertex[x=-4,y=2.5,L={$j$}]{p21}
\Vertex[x=-2,y=2.5,L={$j+1$}]{p22}
\Vertex[x=0,y=2.5,L={$j+i-1$}]{p2i}
\Vertex[x=2,y=2.5,L={$j+n-1$}]{p2n}
\Vertex[x=-1,y=2.5,L={$\dots$}]{dd3}
\Vertex[x=1,y=2.5,L={$\dots$}]{dd4}

\Vertex[x=-2,y=4,L={$p3$}]{p3}
\Vertex[x=-4,y=5,L={$1$}]{p31}
\Vertex[x=-3,y=5,L={$2$}]{p32}
\Vertex[x=-2,y=5,L={$\dots$}]{dd5}
\Vertex[x=-1,y=5,L={$q$}]{p3q}

\Edge[label={$1$}](p1)(p11)
\Edge[label={$2$}](p1)(p12)
\Edge[label={$j$}](p1)(p2)
\Edge[label={$m$}](p1)(p1m)

\Edge[label={$1$}](p2)(p21)
\Edge[label={$2$}](p2)(p22)
\Edge[label={$i$}](p2)(p2i)
\Edge[label={$n$}](p2)(p2n)

\Edge[label={$1$}](p3)(p31)
\Edge[label={$2$}](p3)(p32)
\Edge[label={$q$}](p3)(p3q)

\Edge[color=black,style={dashed},label=$\circ_{j}$](p2i)(p3)

\end{tikzpicture}
\end{center}
\end{enumerate}

In short, the compositions satisfy the same branching rules than the trees together with a correct relabeling of the leaves.\\ \\
Let us give one of the simplest example of operads. Consider a set $\bf S$ and the set of the functions ${\rm Map}_{\bf S}(n)$ from $S^n$ to $S$. The set ${\rm Map}_{\bf S}:=\bigcup_{n\in\mathbb N}{\rm Map}_{\bf S}(n)$  endowed with the classical composition defined by
\[
\left(f\circ \left(g_1,\dots, g_n\right)\right)\left(s_{11},\dots,s_{1k_1},\dots, s_{n,1},\dots,s_{n,k_n}\right):=f\left(g_1\left(s_{11},\dots,s_{1k_1}\right),\dots,g\left(s_{n1},\dots,s_{nk_n}\right)\right)
\]
for each $f\in {\rm Map}_{\bf S}(n)$, $g_i\in{\rm Map}_{\bf S}(k_i)$, $s_{ij}\in S$, is an operad.
\subsection{Free operad, morphisms, suboperads, quotients \emph{etc.}}
The definition of the morphism induces the existence of operads with universal properties called {\it free operads}. Let $G=(G_k)_k$ be a collection of sets, the set ${\rm Free}_G(n)$ is the set of planar rooted trees with $n$ leaves with labeled nodes where nodes with $k$ branches are labeled by the elements of $G_k$. The free operad on $G$ is obtained by endowing the set ${\rm Free}_G=\bigcup_{n\in \mathbb N }{\rm Free}_G(n)$ with the composition ${\bf p}_1\circ_i {\bf p}_2$ which consists in grafting the $i$th leaf of ${\bf p}_1$ with the root of ${\bf p}_2$. Note that ${\rm Free}_G$ contains a copy of $G$ which is the set of the trees with only one inner node (the root) labeled by elements of $G$; for simplicity we will identify it with $G$. The universality means that for any map $\varphi:G\rightarrow {\goth P}$ it exists a unique morphism of operad $\phi: {\rm Free}_G\rightarrow {\goth P}$ such that $\phi(g)=\varphi(g)$ for each $g\in G$.\\
Consider the equivalence relation $\equiv_\phi$ on ${\rm Free}_G$ defined by ${\bf p}_1\equiv_\phi{\bf p}_2$ if and only if $\phi({\bf p}_1) =\phi({\bf p}_2)$. Obviously this relation is compatible with the composition in the sense that ${\bf p}_1\equiv_\phi {\bf p}'_1$ and ${\bf p}_2\equiv_\phi {\bf p}'_2$ implies ${\bf p}_1\circ_i{\bf p}_2\equiv_\phi {\bf p}'_1\circ_i{\bf p}'_2$. Hence, the operad $\goth P$ can be defined as a {\it quotient} (up to an isomorphism) of the free operad: ${\goth P}={\rm Free}_G/_{\equiv_\phi}$. More generally, when an equivalence relation $\equiv$ on the elements of an operad $\goth P$ is compatible with the composition, this induces a structure of operad on the set ${\goth P}/_{\equiv} $; the resulting operad is called the {\it quotient} of $\goth P$ by $\equiv$.\\

The dual notion of quotient is the notion of suboperad whose definition is very classical: a suboperad of an operad ${\goth P}$ is a subset of ${\goth P}$ which contains $\bf 1$ and is stable by composition.\\ 

Consider a set $\bf S$ together with an action of an operad $\goth P$. That is: for each ${\bf p}\in{\goth P}(n)$ we define a map ${\bf p}: {\bf S}^n\rightarrow {\bf S}$. We will say that $\bf S$ is a $\goth P$-module if the action of $\goth P$ is compatible with the composition in the following sense: for each ${\bf p}_1\in {\goth P}(m)$, ${\bf p}_2\in {\goth P}(n)$, $1\leq i\leq m$, $s_1,\dots, s_{m+n-1}\in {\bf S}$ one has:
\[
 {\bf p}_1(s_1,\dots,s_{i-1},{\bf p_2}(s_i,\dots,s_{i+n-1}),s_{i+n},\dots, s_{m+n-1})=({\bf p}_1\circ_i{\bf p}_2)(s_1,\dots,s_{m+n-1}),
\]
\emph{i.e.} there is a morphism of operads from ${\rm Map}_{\bf S}$ to ${\goth P}$.

\section{Multitildes operad\label{multitilde}}

  Multitildes operators have been introduced by Caron~\emph{et al.}~\cite{CCM11e} in order to increase the number of regular languages represented by an expression of a fixed width. The only operation involved in the computation of the languages they denote is the addition of the empty word in several catenation factors in a catenation product. This operation is not very interesting as long as it is considered as unary, since for every regular expression $E$, the addition of the empty word in $L(E)$ is denoted by the expression $E+\varepsilon$, the width of which is the same as $E$. Multitildes  extend it to $k$-ary operators allowing to define new regular expressions, the EMTREs (Extended to Multitilde Regular Expression) \cite{CCM11e}. 
 One of the main interest of multitilde operators is that, for any simple regular expression  with no star, there exists an equivalent EMTRE with only one $k$-ary operator which is a multitilde. We will see that this equivalence involves a natural composition which defines a structure of operad. 

In this section, we first recall the main definitions and results about multitildes then we describe in more details the operad structure acting on languages. 

\subsection{Extended to multitilde regular expression}

   Let $\Sigma$ be an  alphabet and  $\varepsilon$ be the empty word. A \emph{regular expression} $E$ over $\Sigma$ is inductively defined by $E=\emptyset$, $E=\varepsilon$, $E=a$, $E=(F\cdot G)$, $E=(F+G)$, $E=(F^*)$ where $a\in\Sigma$ and $F$, $G$ are regular expressions. The \emph{language denoted by} the expression $E$ is inductively defined by $L(\emptyset)=\emptyset$, $L(\varepsilon)=\{\varepsilon\}$, $L(a)=\{a\}$, $L(F+G)=L(F)\cup L(G)$, $L(F\cdot G)=L(F)\cdot L(G)$, $L(F^*)=L(F)^*$ where $a\in\Sigma$ and $F$, $G$ are regular expressions. 

  Let $i,j,n$ be three positive integers. Let $(E_1,\ldots,E_n)$ be a list of $n$ expressions. The catenation $E_i\cdot E_{i+1}\cdots E_j$ is denoted by $E_{i\cdots j}$ and for convenience the language $L(E_{i\cdots j})$ equals $\{\varepsilon\}$ when $i>j$. A list $(E_i,E_{i+1},\ldots,E_j)$ is denoted by $E_{1,n}$. Given a set $S$, we denote by $\Card{S}$ the number of elements of $S$. Suppose that $i\leq j$. We set $\Intervalle{i}{j}=\{i,i+1,\ldots,j-1,j\}$, $\Intervalle{i}{f}^2_{\leq}=$ $\{(k,k')\mid k,k'\in\Intervalle{i}{f}\text{ and }k\leq k'\}$ and $\mathcal{S}_n={\cal P}(\Intervalle{1}{n}^2_\leq)$ the set of the subsets of $\Intervalle{1}{n}^2_\leq$. The set of indices of $S\in{\cal S}_n$ is $I_S=\Intervalle{1}{\Card{S}}$ and it holds $S=\{(i_k,f_k)\}_{k\in I_S}$, with for all $k\in I_S$, $(i_k,f_k)\in\Intervalle{1}{n}^2_\leq$.

A multitilde of arity $n$ is a formal symbol  $\Ppn T$ where $T\subset {\cal S}_n$. Remark for convenience, we use the same symbol to denote several operations with different arities. For instance, the multitilde $\Ppn{\{(1,3)\}}$ should be $n$-ary for any $n\geq 3$. 
 These symbols together with the regular operations define   a new family of expressions. 
  
  \begin{definition}[\cite{CCM11e}]\label{def ERABT}
    An \emph{Extended to Multitilde Regular Expression} (\textbf{EMTRE}) over an alphabet $\Sigma$ is inductively defined by:

    \centerline{
      \begin{tabular}{r@{ }l@{\ \ }r@{ }l@{\ \ }r@{ }l}
	$E$ & $=\emptyset$,      &  &  & $E$ & $=a$, \\
	$E$ & $=(F+G)$,    &  $E$ & $=(F\cdot G)$, &  $E$ & $=(F^{*})$,\\	              
      \end{tabular}
    }
    
    \centerline{$E=(\Pp{T}{E_{1,n}})$,}
    
    
    \centerline{
	with $a \in\Sigma$, $F$ and $G$  are two EMTREs,
	     $E_{1,n}$ a list of EMTREs and $T$ a list in $\mathcal{S}_n$.
    }
  \end{definition}

Now let us recall how to extend the notion of languages denoted by  regular expressions to EMTREs, by defining the action of multitilde operators.\\ 

A set $S\in{\cal S}_n$ is said \emph{free} if and only if for all $(i,f),(i',f')$ in $S$ such that $(i,f)\neq (i',f')$, the condition $\Intervalle{i}{f}\cap \Intervalle{i'}{f'}=\emptyset$ holds. Let $S\in {\cal S}_n$, we denote by ${\cal F}(S)$ the set of all free subsets of $S$.

Let $L_1,\dots,L_n$ be a list of languages over an alphabet $\Sigma$ and let $T$ be a list in $\mathcal{S}_n$. We set $\mathcal{W}_T(L_1,\dots,L_n)=L'_1\cdots L'_n$ where 
\[
 L'_k=\left\{\begin{array}{ll}
\{\varepsilon\}&\mbox{ if } k\in\bigcup_{(i,f)\in T} \Intervalle{i}{f}\\
L_k&\mbox{ otherwise.}
\end{array}\right.
\]

  \begin{definition}
    Let $L_{1,n}$ be a list of languages over an alphabet $\Sigma$. Let $T$ be a list in $\mathcal{S}_n$. The language   $\Pp{T}{L_{1,n}}$ is defined by: 

 $$\Pp{T}{L_{1,n}}:=\bigcup_{S\in {\cal F}(T)}\mathcal{W}_S(L_{1,n}).$$ \end{definition}
  
Let $E_{1,n}$ be a list of EMTREs, the language denoted by  $\Pp{T}{E_{1,n}}$ where $T\in \mathcal{S}_n$ is inductively defined by 
$$L(\Pp{T}{E_{1,n}}):=\Pp{T}{L(E_1),\ldots, L(E_n)}.$$ 
As usual two EMTREs $E$ and $F$ are said {\it equivalent} is they denote the same language; we write $E\equiv F$.\\

\noindent Note all these definitions are slight rewordings to those given in \cite{CCM11e}.

\noindent  Every regular expression  $E$ that does not use any star operator can be turned into an equivalent expression $F=\Pp{T}{e_1,\ldots,e_n}$ where $\Pp{T}\in\mathcal{S}_n$ and $\forall k\in\Intervalle{1}{n}$, $e_k\in\Sigma\cup\{\emptyset\}$. Conversion formulas are given Table~\ref{tab regl modif mtb}.
  
  \begin{table}[H]
    \centerline{
      \begin{tabular}{c@{\ }c@{\ }c}
	    $\emptyset$ & $\equiv$ & $\Pp{\emptyset}{\emptyset}$\\
	    $\varepsilon$ & $\equiv$ & $\Pp{(1,1)}{\emptyset}$\\
	    $a$ & $\equiv$ & $\Pp{\emptyset}{a}$\\
      \end{tabular}
      \hspace{5em}
      \begin{tabular}{c@{\ }c@{\ }c}
	    $E'_1+E'_2$ & $\equiv$ & $\Pp{(1,2),(2,3)}{E'_1,\emptyset,E'_2}$\\
	    $E'_1\cdot E'_2$ & $\equiv$ & $\Pp{\emptyset}{E'_1,E'_2}$\\
      \end{tabular}
    }
    \caption{Multitildes conversion.}
    \label{tab regl modif mtb}
  \end{table}

  One of the authors has shown in~\cite{Mig10} how to compose multitildes operators preserving languages:
  
  \centerline{$L\left(\Pp{T}{E_1,\ldots,E_{k-1},\Pp{T'}{E'_1,\ldots,E'_{k'}},E_{k+1},\ldots,E_n}\right)=L\left(\Pp{T\circ_k T'}{E_1,\ldots,E_{k-1},E'_1,\ldots,E'_{k'},E_{k+1},\ldots,E_n}\right)$.}
We do not recall the original definition here but we   describe this operation using new operators, namely $\Decn$ and $\Shiftn$ in the sequel of this section.
We   show that this composition endows the set of multitildes with a structure of operad.

\subsection{The operators $\Decn$ and $\Shiftn$}  

  The first operators to be defined are the $\Decn$ operators, parametrized by any integer $k$. Any operator $\Decnk{k}$ will increase both elements of a couple $(x,y)$ of integers by $k$:
  
  \centerline{$\Dec{k}{x,y}:=(x+k,y+k).$}

  These operators are commutative ones:    
  \begin{equation}\label{CommutationDec}
    \Dec{k}{\Dec{\ell}{x,y}} :=\Dec{\ell}{\Dec{k}{x,y}}
  \end{equation}
  
  Indeed, both of the expressions result in the same couple:
  
  \begin{equation}\label{CompositionDec}
    \Dec{k}{\Dec{\ell}{x,y}}=(x+k+\ell ,y+k+\ell)=\Dec{\ell}{\Dec{k}{x,y}}
  \end{equation}

  The second family of operators is the set of $\Shiftnk{k,n}$ operators for any integers $n,k$, which shift elements in couples by inserting $n$ elements in position $k$ (see Example~\ref{ex explic shift} for details):

  \centerline{
    $\Shift{k,n}{x,y}:=\left\{
      \begin{array}{ll}
	    (x,y) & \text{ if } y<k, \\
	    (x,y+n-1)& \text{ if } x\leq k\leq y, \\
	    \Dec{n-1}{x,y}& \text{ otherwise.} \\
      \end{array}
    \right.$
  }
  
  \begin{example}\label{ex explic shift}
    Let $n$ be a positive integer, and $x\leq y$ be two integers in $\Intervalle{1}{n}$. The operators $\Shiftnk{k,n}$ transforms the couple $(x,y)$ performing an insertion of $n$ elements at position $k$ presented in Figure~\ref{fig explic shift}. After an insertion of $6$ elements in position $5$ (using the operator $\Shiftnk{5,6}$), couples $(1,3)$, $(3,7)$ and $(7,8)$ are respectively transformed into $(1,3)$, $(3,12)$ and $(12,13)$.
  \end{example}
  
  \begin{figure}[H]
    \begin{minipage}{0.4\linewidth}
      \begin{tikzpicture}[scale=0.5]
	    \draw (1,0) -- (10,0);
	    \draw[thick] (1,0.3) -- (4,0.3);
	    \draw[thick] (7,0.3) -- (9,0.3);
	    \draw[thick] (3,0.6) -- (8,0.6);
	    \foreach \x in {1,...,10}
	      \draw (\x,-2pt) -- (\x,2pt);
	    \foreach \x in {1,...,4}
	      \draw (\x ,2pt)++(0.5,0) node[anchor=north] {$\x$};
	    \draw (5 ,2pt)++(0.5,0) node[anchor=north] {\underline{5}};
	    \foreach \x in {6,...,9}
	      \draw (\x ,2pt)++(0.5,0) node[anchor=north] {$\x$};
	  \end{tikzpicture}
    \end{minipage}
    \hfill
    \begin{minipage}{0.55\linewidth}
      \begin{tikzpicture}[scale=0.5]
	    \draw (1,0) -- (15,0);
	    \draw[thick] (1,0.3) -- (4,0.3);
	    \draw[thick] (12,0.3) -- (14,0.3);
	    \draw[thick] (3,0.6) -- (13,0.6);
	    \foreach \x in {1,...,15}
	      \draw (\x,-2pt) -- (\x,2pt);
	    \foreach \x in {1,...,4}
	      \draw (\x ,2pt)++(0.5,0) node[anchor=north] {$\x$};
	    \foreach \x in {5,...,10}
	      \draw (\x ,2pt)++(0.5,0) node[anchor=north] {\underline{$\x$}};
	    \foreach \x in {11,...,14}
	      \draw (\x ,2pt)++(0.5,0) node[anchor=north] {$\x$};
      \end{tikzpicture}
    \end{minipage}
    \caption{The $\Shiftnk{5,6}$ operator.}
    \label{fig explic shift}
  \end{figure}
  
  For convenience, definition of $\Shiftnk{k,n}$ operators can be rewritten using the characteristic function $\mathbbm{1}_E(x)$ defined for any integer set $E$ and any integer $x$ by:
  
  \centerline{
      $\mathbbm{1}_E(x)=
        \left\{
          \begin{array}{l@{\ }l}
            1 & \text{ if }x\in E,\\
            0 & \text{ otherwise.}\\
          \end{array}
        \right.$
  }
  
  Since we will only use interval, we will denote by $\mathbbm{1}_k(x)$ for any integer $k$ the function $\mathbbm{1}_{\Intervalle{k}{\infty}}(x)$.
  
  \begin{lemma}
    Let $k,n,x,y$ be four integers. Then:
    
  \centerline{
    $\Shift{k,n}{x,y}=(x+\mathbbm{1}_{k+1}(x)\times(n-1),y+\mathbbm{1}_{k}(y)\times(n-1))
    $.
  }
  \end{lemma}
  
  Next lemma presents the commutation of $\Decn$ and $\Shiftn$.

  \begin{lemma}\label{CompositionDecShift}
    Let $x,y,i,m,n$ be five integers with $x\leq y$. Then:
    
    \centerline{$\Dec{n}{\Shift{i,m}{x,y}}=\Shift{i+n,m}{\Dec{n}{x,y}}.$}
  \end{lemma}
  \begin{proof}
    From formulas for $\Decn$ and $\Shiftn$ operators:
  
    \centerline{
      \begin{tabular}{l@{\ }l}
        $\Dec{n}{\Shift{i,m}{x,y}}$ & $=\Dec{n}{x+\mathbbm{1}_{i+1}(x)\times(m-1),y+\mathbbm{1}_{i}(y)\times(m-1)}$\\
        & $=(x+n+\mathbbm{1}_{i+1}(x)\times(m-1),y+n+\mathbbm{1}_{i}(y)\times(m-1))$\\
        & $=(x+n+\mathbbm{1}_{i+n+1}(x+n)\times(m-1),y+n+\mathbbm{1}_{i+n}(y+n)\times(m-1))$\\
        & $=\Shift{i+n,m}{x+n,y+n}$\\
        & $=\Shift{i+n,m}{\Dec{n}{x,y}}$\\
      \end{tabular}
    }
    
  \end{proof}
  
  The $\Shiftn$ operators commute in a particular way, described as follows:

  \begin{lemma}\label{CommutationShift}
    Let $x,y,m,n,k,\ell$ be six integers with $x\leq y$ and $k<\ell$. Then:

    \centerline{$\Shift{k,m}{\Shift{\ell,n}{x,y}}=\Shift{\ell+m-1,n}{\Shift{k,m}{x,y}}.$}
  \end{lemma}
  \begin{proof}
    From formulas for $\Shiftn$ operators:
    
    \centerline{
      \begin{tabular}{l@{\ }l@{\ }l}
        $\Shift{k,m}{\Shift{\ell,n}{x,y}}$ & $=$ & $\Shift{k,m}{x+\mathbbm{1}_{\ell+1}(x)\times(n-1), y+\mathbbm{1}_{\ell}(y)\times(n-1)}$\\    
        & $=$ & $(x',y')$
      \end{tabular}
    }
    
    \centerline{
      \begin{tabular}{l@{\ }l}
        with & $x'=x+\mathbbm{1}_{\ell+1}(x)\times(n-1)+\mathbbm{1}_{k+1}(x+\mathbbm{1}_{\ell+1}(x)\times(n-1))\times(m-1)$\\
        and & $y'=y+\mathbbm{1}_{\ell}(y)\times(n-1)+\mathbbm{1}_{k}(y+\mathbbm{1}_{\ell}(y)\times(n-1))\times(m-1)$.\\
      \end{tabular}
    }
    
    Since $k<\ell$, then
    
    \centerline{
      \begin{tabular}{l@{\ }l}
        & $\mathbbm{1}_{\ell+1}(x)$
     $=$
    $\mathbbm{1}_{\ell+m}(x+m-1)$
    $=$
    $\mathbbm{1}_{\ell+m}(x+\mathbbm{1}_{k+1}(x)\times(m-1))$\\
    and
    & $\mathbbm{1}_{k+1}(x)$
    $=$
    $\mathbbm{1}_{k+1}(x+\mathbbm{1}_{\ell+1}(x)\times(n-1))$.\\
      \end{tabular}
    }
    
    Consequently,
    
    \centerline{
      \begin{tabular}{l@{\ }l}
      & $x'=
    x+\mathbbm{1}_{k+1}(x)\times(m-1)+\mathbbm{1}_{\ell+m}(x+\mathbbm{1}_{k+1}(x)\times(m-1))\times(n-1),$\\
    and & $y'=y+\mathbbm{1}_{k}(x)\times(m-1)+\mathbbm{1}_{\ell+m-1}(y+\mathbbm{1}_{k}(y)\times(m-1))\times(n-1)
    $.\\
      \end{tabular}
    }
    
    Finally,
    
    \centerline{
      \begin{tabular}{l@{\ }l}
        $(x',y')$ & $=\Shift{\ell+m-1,n}{(x+\mathbbm{1}_{k+1}(x)\times(m-1),y+\mathbbm{1}_{k}(y)\times(m-1))}$  \\
        & $=\Shift{\ell+m-1,n}{\Shift{k,m}{x,y}}.$
      \end{tabular}
    }
    
  \end{proof}

  Finally, another combination of $\Shiftn$ operators:

  \begin{lemma}\label{CompositionShift}
    Let $i,n,k,m,x,y$ be six integers such that $x\leq y$ and  $0\leq i< n$. Then:

    \centerline{$\Shift{k+i,m}{\Shift{k,n}{x,y}}=\Shift{k,m+n-1}{x,y}$.}
  \end{lemma}
  \begin{proof}
    From formulas for $\Shiftn$ operators:
    
    \centerline{
      \begin{tabular}{l@{\ }l}
        $\Shift{k+i,m}{\Shift{k,n}{(x,y)}}$ & $=\Shift{k+i,m}{x+\mathbbm{1}_{k+1}(x)\times(n-1),y+\mathbbm{1}_{k}(y)\times(n-1)}$\\
        & $=(x',y')$\\
      \end{tabular}
    }  
    
    \centerline{
      \begin{tabular}{l@{\ }l}
        with & $x'=x+\mathbbm{1}_{k+1}(x)\times(n-1)+\mathbbm{1}_{k+i+1}(x+\mathbbm{1}_{k+1}(x)\times(n-1))\times(m-1)$\\
        and & $y'=y+\mathbbm{1}_{k}(y)\times(n-1)+\mathbbm{1}_{k+i}(y+\mathbbm{1}_{k}(y)\times(n-1))\times(m-1)$.\\
      \end{tabular}
    }
    
    Since $0\leq i < n$, then
    
    \centerline{
      $\mathbbm{1}_{k+1}(x)=\mathbbm{1}_{k+i+1}(x+\mathbbm{1}_{k+1}(x)\times(n-1))$.
    }
    
    Finally,
    
    \centerline{
      \begin{tabular}{l@{\ }l}
        $(x',y')$ & $=(x+\mathbbm{1}_{k+1}(x)\times(m+n-1),y+\mathbbm{1}_{k}(y)\times(m+n-1))$  \\
        & $=\Shift{k,m+n-1}{x,y}$.
      \end{tabular}
    }
    
  \end{proof}

  The definition domain of $\Shiftn$ operators can be extended to sets of couples by the following way. Let $T\subset\Intervalle{1}{n}^2_\leq$ for any integer $n$:

  \centerline{$\Shift{k,n}{T}=\{\Shift{k,n}{(x,y)}:(x,y)\in T\},\,\Dec{k}{T}=\{\Dec{k}{(x,y)}:(x,y)\in T\}$.}

  From an algebraic point of view, operators $\Shiftn$ and $\Decn$ are linearly extended. Relations are still valid:

  \begin{proposition}\label{relationsSD}
    Let $k,\ell,n,i,m$ be five integers, and let $T$ be a set of couples $(x,y)$ with $x\leq y$. Then: 

    \begin{enumerate}
      \item  $\Dec{k}{\Dec{\ell}{T}} :=\Dec{\ell}{\Dec{k}{T}}$,
      \item  $\Dec{n}{\Shift{i,m}{T}}=\Shift{i+n,m}{\Dec{n}{T}}$,
      \item $\Shift{k,m}{\Shift{\ell,n}{T}}=\Shift{\ell+m-1,n}{\Shift{k,m}{T}}$ if $k\leq\ell$.
      \item $\Shift{k+i,m}{\Shift{k,n}{T}}=\Shift{k,m+n-1}{T}$ if $0\leq i< n$.
    \end{enumerate}
  \end{proposition}

\subsection{The Operad of multitildes}

  Multitildes are $k$-ary operators for any integer $k>0$ represented by a subset of $\Intervalle{1}{n}^2_\leq$. The set of $k$-ary multitildes is denoted by $\mathcal{T}_k$. The set of multitildes (\emph{i.e.} $\bigcup_{k\in\mathbb{N}} \mathcal{T}_k$) is denoted by $\mathcal{T}$.

  Let $k,n$ be two integers such that $1\leq k\leq n$. The partial composition $\circ_k$ of multitildes can be expressed from $\Shiftn$ and $\Decn$ operators as follows:

    \centerline{
      $\circ_k:\left\{
	\begin{array}{ccc}
	  {\cal T}_m\otimes{\cal T}_n & \rightarrow & {\cal T}_{n+m-1}\\
	  T_1 \circ_k T_2& = & \Shift{k,n}{T_1}\cup\Dec{k-1}{T_2}
	\end{array}
      \right.$
    }

Note this composition is straightforwardly the same than those defined by one of the authors in \cite{Mig10}.
  It satisfies the property of commutativity and associativity defining an operad.

  \begin{proposition}[\it Commutativity]\label{CommutationTilde}
    Let $m,n,p, k,\ell$ be five integers such that $1\leq k< \ell\leq m$. Let $T_1\in{\cal T}_{m}$, $T_2\in{\cal T}_{n}$ and $T_3\in{\cal T}_{p}$ be three multitildes. Then:

    \centerline{$(T_1\circ_\ell T_2)\circ_kT_3=(T_1\circ_kT_3)\circ_{\ell+p-1}T_2$.}
  \end{proposition}
  \begin{proof}
    From the definition of composition:

    \centerline{$(T_1\circ_\ell T_2)\circ_k T_3=\Shift{k,p}{\Shift{\ell,n}{T_1}}\cup
    \Shift{k,p}{\Dec{\ell-1}{T_2}}\cup \Dec{k-1}{T_3}.$}

    Following Proposition~\ref{relationsSD} on the sets of couples computed in the right member:

    \begin{enumerate}
      \item From relation~(3): $\Shift{k,p}{\Shift{\ell,n}{T_1}}=\Shift{\ell+p-1,n}{\Shift{k,p}{T_1}}$;

      \item Since $k< \ell$, from relation~(2): 
      
      \centerline{$\Shift{k,p}{\Dec{\ell-1}{T_2}}=\Dec{\ell-1}{\Shift{k-\ell+1,p}{T_2}}=\Dec{\ell+p-2}{T_2}$.}
    \end{enumerate}

    Consequently: $(T_1\circ_\ell T_2)\circ_k T_3= \Shift{\ell+p-1,n}{\Shift{k,p}{T_1}}\cup
    \Dec{\ell+p-2}{T_2}\cup \Dec{k-1}{T_3}$.
    
    Since $k<\ell$, the operator $\Shiftn_{\ell+p-1,n}$ does not modify the set $\Dec{k-1}{T_3}$ of couples the elements of couples of which are upper-bounded by $k+p-1<\ell+p-1$. Hence: $\Dec{k-1}{T_3}=\Shift{\ell+p-1,n}{\Dec{k-1}{T_3}}$.
    
    Then: $(T_1\circ_\ell T_2)\circ_k T_3=\Shift{\ell+p-1,n}{\Shift{k,p}{T_1}}\cup\Shift{\ell+p-1,n}{\Dec{k-1}{T_3}}
    \cup \Dec{\ell+p-2}{T_2}$.

    Finally: $(T_1\circ_\ell T_2)\circ_k T_3=(T_1\circ_kT_3)\circ_{\ell+p-1}T_2$.
    
  \end{proof}

  The associativity of $\circ_k$ is shown by the following proposition.

  \begin{proposition}\label{AssociativiteTilde}
    Let $m,n,p,i,k$ be five integers such that $i<n$. Let $T_1\in{\cal T}_m$, $T_2\in{\cal T}_n$ and $T_3\in{\cal T}_p$ be three multitildes. Then:

    \centerline{$T_1\circ_k(T_2\circ_iT_3)=(T_1\circ_kT_2)\circ_{k+i-1}T_3$.}
  \end{proposition}
  \begin{proof}
    According to definition of $\circ_i$:

    \centerline{
      $\begin{array}{rcl}
	T_1\circ_k(T_2\circ_i T_3) & = & T_1\circ_k\left(\Shift{i,p}{T_2}\cup\Dec{i-1}{T_3}\right)\\
	& = & \Shift{k,n+p-1}{T_1}\cup\Dec{k-1}{\Shift{i,p}{T_2}}\cup\Dec{k-1}{\Dec{i-1}{T_3}}.
      \end{array}$
    }

    Following Proposition~\ref{relationsSD} on the sets of couples computed in the right member:

    \begin{enumerate}
      \item from relation~(1): $\Dec{k-1}{\Dec{i-1}{T_3}}=\Dec{k+i-2}{T_3}$;

      \item from relation~(2): $\Dec{k-1}{\Shift{i,p}{T_2}}=\Shift{k+i-1,p}{\Dec{k-1}{T_2}}$;

      \item Since $i<n$, from relation~(4): $\Shift{k,n+p-1}{T_1}=\Shift{k+i-1,p}{\Shift{k,n}{T_1}}$.
    \end{enumerate}

    Consequently: $T_1\circ_k(T_2\circ_i T_3)  =  \Shift{k+i-1,p}{\Shift{k,n}{T_1}} \cup \Shift{k+i-1,p}{\Dec{k-1}{T_2}}\cup\Dec{k+i-2}{T_3}$.

    The right side of the expression is the definition of $\circ_{k+i-1}$ and finally:

    \centerline{
      $\begin{array}{rcl}
	T_1\circ_k(T_2\circ_i T_3) & = & (\Shift{k,n-1}{T_1}\cup\Dec{k-1}{T_2})\circ_{k+i-1} T_3\\
	& = & (T_1\circ_k T_2)\circ_{k+i-1}T_3.
      \end{array}$
    }
    
  \end{proof}
  
  From the partial composition $\circ_k$, ones can define the composition $\circ$ as follows:
  
  \centerline{
      $\circ:\left\{
	\begin{array}{ccc}
	  {\cal T}_m \otimes {\cal T}_{n_1}\otimes \cdots \otimes {\cal T}_{n_m}  & \rightarrow & {\cal T}_{n_1+\cdots +n_m}\\
	  T\circ (T'_1,\ldots,T'_m) & = &   (\cdots((T_1 \circ_m T'_m)\circ_{m-1} T'_{m-1})\cdots)\circ_1 T'_1.
	\end{array}
      \right.$
    }
  
  \begin{theorem}\label{TildeOperad}
    The structure $(\mathcal{T},\circ)$ is an operad.
  \end{theorem}
  
\section{Actions on languages\label{SecActLan}}

In this section, we recover the construction of language associated to a multitildes expression (\cite{CCM11e}) by means of an operad based on boolean vectors. We show the connection between this operad and the operad $\mathcal T$.
\subsection{An operad of boolean vectors and its action on languages}

Denote by ${\mathbb B}=\{0,1\}$ and $\mathcal B_n=\mathcal P(\mathbb B^n)$. We endow the set $\mathcal B=\bigcup_{n\in\N}\mathcal B_n$ with a structure of operad whose composition is   as follows:\\
 Let $E\in \mathcal B_m$ and $F\in \mathcal B_n$, the set $E\circ_k F$ is defined by
\[\begin{array}{rcl}
E\circ_k F&:=&\{[e_1,\dots,e_{k-1},e_kf_1,\dots,e_kf_n,e_{k+1},\dots,e_m]:[e_1,\dots,e_m]\in E,\ [f_1,\dots,f_n]\in F\}.\end{array}
\]
\begin{proposition}
The structure $(\mathcal B,\circ)$ is an operad.
\end{proposition}
\begin{proof}
 The set ${\bf 1}:=\{[1]\}$ is the identity for the composition $\circ_k$. Indeed,
\[
{\bf 1}\circ_1 F=F
\]
and
\[\begin{array}{rcl}
E\circ_k\{1\}
&:=&\{[e_1,\dots,e_{k-1},0,e_{k+1},\dots,e_m]:[e_1,\dots,e_{k-1},0,e_{k+1},\dots,e_m]\in E\}
\cup \\&& \{[e_1,\dots,e_{k-1},1,e_{k+1},\dots,e_m]:[e_1,\dots,e_{k-1},1,e_{k+1},\dots,e_m]\in E\}=E.\end{array}
\]
It remains to prove the two associativity rules:
\begin{enumerate}
\item Let $E\in \mathcal B_m$, $F\in {\mathcal B_n}$, $G\in\mathcal B_p$ and $1\leq j<i\leq m$. One has
\[
\begin{array}{rcl}
(E\circ_i F)\circ_j G&=&\{[e_1,\dots, e_{j-1},e_jg_1,\dots,e_jg_n,e_{j+1},\dots, e_{i-1},e_if_1,\dots, e_if_n,e_{i+1},\dots,e_n]\\&&:[e_1,\dots, e_m]\in E, [f_1,\dots, f_n]\in F, [g_1,\dots, g_p]\in G\}\\
&=& \{[e_1,\dots,e_{j-1},e_jg_1,\dots,e_jg_p,e_{j+1},\dots,e_m]:[e_1,\dots, e_m]\in E,\\&& [g_1,\dots, g_p]\in G \}\circ_{j+p-1} F\\
&=&(E\circ_j G)\circ_{j+p-1} F.
\end{array}
\]
\item Let $E\in \mathcal B_m$, $F\in {\mathcal B_n}$, $G\in\mathcal B_p$, $1\leq i\leq n$ and $1\leq j\leq m$. One has
\[
\begin{array}{rcl}
E\circ_j(F\circ_i G)&=&\{[e_1,\dots,e_{j-1},e_jf_1,e_jf_ig_1,\dots,e_jf_ig_p,e_jf_{i+1},\dots,
e_jf_n,e_{j+1},\dots,e_m]:\\
&& [e_1,\dots,e_m]\in E, [f_1,\dots,f_n]\in F, [g_1,\dots,g_p]\in G\}\\
&=& (E\circ j F)\circ_{i+j-1} G.
\end{array}
\]

\end{enumerate}
This proves that $\mathcal B$ is an operad.
\end{proof}
We define the action of $\mathcal B_n$ on the languages by
\[
E(L_1,\dots,L_n)=\bigcup_{[e_1,\dots,e_n]\in E}L_1^{e_1}\cdot L_2^{e_2}\cdots L_n^{e_n}.
\]
for each $E\in\mathcal B$.
\begin{proposition}
The sets $2^{\Sigma^*}$ (the set of the languages over $\Sigma$) and $\mathrm{Reg}(\Sigma^*)$ (the set of regular languages over $\Sigma$) are both $\mathcal B$-modules.
\end{proposition}
\begin{proof}
Let us show first the compatibility of the action previously defined with the composition.
One has
\[\begin{array}{rcl}
E(L_1,\dots,L_{i-1},F(L_i,\dots,L_{i+n-1}),L_{i+n},\dots,L_{m+n-1})&=&\displaystyle
\bigcup_{[e_1,\cdots,e_m]\in E}L_1^{e_1}\cdots L_{i-1}^{e_{i-1}}
\left(L_i^{f_1}\cdots L_{i+n-1}^{f_n}\right)^{e_i}L_{i+n}^{e_i+1}\cdots L_{m+n-1}^{e_m}.\end{array}
\]
Since $e_i\in\{0,1\}$,
\[\begin{array}{rcl}
E(L_1,\dots,L_{i-1},F(L_i,\dots,L_{i+n-1}),L_{i+n},\dots,L_{m+n-1})&=&\displaystyle
\bigcup_{[e_1,\dots,e_m]\in E,\atop
[f_1,\dots,f_n]\in F}L_1^{e_1}\cdots L_{i-1}^{e_{i-1}}L_i^{e_if_1}\cdots L_{i+n-1}^{e_if_n}L_{i+n}^{e_{i+1}}\cdots L_{m+n-1}^{e_n}\\
&=& \displaystyle\bigcup_{[g_1,\dots,g_{n+m-1}]\in E\circ_i F}L_1^{g_1}\cdots L_{n+m-1}^{g_{n+m-1}}\\
&=& \displaystyle E\circ_i F(L_1,\dots,L_{n+m-1}).\end{array}\]
Hence, $2^{\Sigma^*}$ is a $\mathcal B$-module. Note a finite union of catenation of regular languages is regular, so each operator $E\in\mathcal B$ maps $\mathrm{Reg}(\Sigma^*)$  on $\mathrm{Reg}(\Sigma^*)$. This implies that it is also a $\mathcal B$-module.
\end{proof}
\subsection{Action of $\mathcal T$ on languages}

We define a map $V: \mathcal T_k\rightarrow \mathcal B_k$ by
\[
 V(T)=\{v(S):S\in {\mathcal F(T)}\}
\] 
with
\[
v(S)=(v_1,\dots,v_k)\mbox{ where }v_j=\left\{\begin{array}{ll}0&\mbox{ if }j\in\bigcup_{(x,y)\in S}\llbracket x,y\rrbracket\\
1&\mbox{ otherwise}.\end{array}\right.
\]
\begin{example}
\rm Consider $T=\{(1,2),(2,3),(3,4),(4,4)\}\in\mathcal T_4$. The images of the elements of $\mathcal F(T)$ are 
\[\begin{array}{rcl}
v(\{(1,2),(3,4)\})&=&(0,0,0,0)\\
v(\{(1,2),(4,4)\})&=&(0,0,1,0)\\
v(\{(2,3),(4,4)\})&=&(1,0,0,0)\\
v(\{(1,2)\})&=&(0,0,1,1)\\
v(\{(2,3)\})&=&(1,0,0,1)\\
v(\{(3,4)\})&=&(1,1,0,0)\\
v(\{(4,4)\})&=&(1,1,1,0)\\
v(\emptyset)&=&(1,1,1,1)
\end{array}
\]
Hence 
$$V(T)=\{(0,0,0,0),(0,0,1,0),(1,0,0,0),(0,0,1,1),(1,0,0,1),(1,1,0,0),(1,1,1,0),(1,1,1,1)\}$$
\end{example}
\begin{theorem}
 The map $V$ is a morphism of operads.
\end{theorem}
\begin{proof}
First the image of $\emptyset\in\mathcal T_1$ by $V$ is the set $\bf 1$.\\
Now let us examine the image of a composition. Let $T_1\in \mathcal T_m$, $T_2\in\mathcal T_n$ and $1\leq k\leq m$.
\\ We need the following lemma which explains how the composition modifies the free subsets.
\begin{lemma}
\begin{equation}\label{FT1T2}\begin{array}{rcl}
\mathcal F(T_1\circ_k T_2)&=&\{\Shiftnk{n,k}S:S\in\mathcal F(T_1)\}\cup\\
&&
\{
\Shiftnk{n,k}S\cup\Decnk{k-1}T:S\in\mathcal F(T_1), T\in\mathcal F(T_2),\\&& \forall(x,y)\in S, x\geq k\mbox{ or }y<k-1\}
\end{array}
\end{equation}
\end{lemma}
\begin{proof}
For simplicity, we will denote by $\mathcal R$ the right hand side of equality (\ref{FT1T2}).
Note $\mathcal R\subseteq \mathcal F(T_1\circ_k T_2)$ is obtained remarking that each element of $\mathcal R$ is free.\\
Now let us prove the inclusion $\mathcal F(T_1\circ_k T_2)\subseteq \mathcal R$. Straightforwardly from the definition of $\mathcal F$, we have the inclusion
\[
\mathcal F(T_1\circ_k T_2)\subset\{ \Shiftnk{n,k}S\cup\Decnk{k-1}T:S\in\mathcal F(T_1), T\in\mathcal F(T_2)\}.
\]
Hence, if $R\in \mathcal F(T_1\circ_k T_2)$, there exist $S\in\mathcal F(T_1)$ and $T\in\mathcal F(T_2)$ such that $R=\Shiftnk{n,k}S\cup\Decnk{k-1}T$. If $T=\emptyset$ then $R\in \{\Shiftnk{n,k}S:S\in\mathcal F(T_1)\}\subset \mathcal R$. Suppose $T\neq \emptyset$ and $R\not\in \mathcal R$. Consequently there exist two couples $(x,y+n),(x'+k-1,y'+k-1)\in R$ with $(x,y)\in S$, $x\leq k$, $y>k-1$ and $(x',y')\in T$. Hence,
$\llbracket x,y+n\rrbracket\cap \llbracket x'+k-1,y'+k-1\rrbracket=\llbracket x'+k-1,y'+k-1\rrbracket\neq\emptyset$. This contradicts the fact that $R$ is free, proves the second inclusion and then the equality of the two sets.
\end{proof}
\noindent{\it End of the proof}
From equation (\ref{FT1T2}) we have
\[\begin{array}{rcl}
V(T_1\circ_k T_2)&=&\{v(\Shiftnk{n,k}S):S\in\mathcal F(T_1)\}\cup\\
&&
\{
v(\Shiftnk{n,k}S\cup\Decnk{k-1}T):S\in\mathcal F(T_1), T\in\mathcal F(T_2),\\&& \forall(x,y)\in S, x\geq k\mbox{ or }y<k-1\}
\end{array}
\]
Set $V_1:=\{
v(\Shiftnk{n,k}S\cup\Decnk{k-1}T):S\in\mathcal F(T_1), T\in\mathcal F(T_2), \forall(x,y)\in S, x\geq k\mbox{ or }y<k-1\}$ and $V_2=\{v(\Shiftnk{n,k}S):S\in\mathcal F(T_1)\}$.
We have
\[
V_1=\{
v(\Shiftnk{n,k}S\cup\Decnk{k-1}T):S\in\mathcal F(T_1), T\in\mathcal F(T_2), v(s)_k=1\}.
\]
Hence,
\[\begin{array}{rcl}
V_1&=&\{(s_1,\dots,s_{k-1},s_kt_1,\dots,s_kt_n,s_{k+1},\dots,s_m):(s_1,\dots,s_m)\in V(T_1),\\&& (t_1,\dots,t_n)\in V(T_2), s_k=1\}\end{array}
\]
On another hand $V_2$ splits as
\[
V_2=V^0_2\cup V^1_2
\]
where
\[
V^i_2=\{v(\Shiftnk{n,k}S:S\in\mathcal F(T_1)\mbox{ and } v(S)_k=i\}.
\]
We have $V^1_2\subset V_1$ (it suffices to put $T=\emptyset$ in $V_1$) and
\[\begin{array}{rcl}
V^0_2&=&\{(s_1,\dots,s_{k-1},s_kt_1,\dots,s_kt_n,s_{k+1},\dots,s_m):(s_1,\dots,s_m)\in V(T_1),\\&& (t_1,\dots,t_n)\in V(T_2), s_k=0\}\end{array}.
\]
It follows that $V(T_1\circ_k T_2)=V(T_1)\circ_kV(T_2)$.
\end{proof}
\begin{corollary}\label{cormod}
The sets $2^{\Sigma^*}$ and ${\rm Reg}(\Sigma^*)$ are both ${\mathcal T}$-module.
\end{corollary}
\begin{proof}
It suffices to set $T(L_1,\dots,L_k):=V(T)(L_1,\dots,L_k)$.
\end{proof}

Note the action of $\mathcal T$ on languages matches with the definition of $\Ppn{T}$.
\begin{proposition}
\[
\Pp{T}{L_1,\dots,L_k}=T(L_1,\dots,L_k)
\]
\end{proposition}
\begin{proof}
We have
\[
\mathcal W_S(L_1,\dots,L_k)=L'_1\cdots L'_k
\]
with 
\[
L'_i=\left\{\begin{array}{ll}
\{\epsilon\}&\mbox{ if }i\in\bigcup_{(x,y)\in S}\llbracket x,y\rrbracket\\
L_k&\mbox{ otherwise}.
\end{array}\right.
\]
Hence,
\[
\mathcal W_S(L_1,\dots,L_k)=L_1^{v(S)_1}\cdots L^{v(S)_k}_k.
\]
And the result follows from
\[
\Pp{T}{L_1,\dots,L_k}=\bigcup_{S\in\mathcal F(T)}\mathcal W_S(L_1,\dots,L_k)
=\bigcup_{v\in V(T)}L_1^{v_1}\cdots L_k^{v_k}=T(L_1,\dots,L_k).
\]
\end{proof}

\section{$\mathcal{T}$, $\mathrm{RAS}$ and $\mathrm{POSet}$\label{SecTRAPO}}
 
In this section, we give a combinatorial description of the operad $\mathcal{T}$. In particular we prove that it is isomorphic to an operad whose underlying set is the set $\mathrm{RAS}^\leq$ of reflexive and (necessarily) antisymmetric subrelations of the order $\leq$ on $\N\setminus \{0\}$.
In our context, we define a quotient $\mathrm{POSet}^\leq$ of $\mathrm{RAS}^\leq$ whose elements are indexed by partial ordered sets. This construction is based on the transitive closure  equivalence relation on $\rm RAS$; the equivalent classes are indexed by POSets. We prove that the operad $\mathrm{POSet}^\leq$ is isomorphic to a quotient of $\mathcal{T}$ which is   compatible in a natural way with the action on languages.

\subsection{From $\mathcal{T}$ to $\mathrm{RAS}^\leq$}

  Since multitildes are defined by a set of couples $(x,y)$ such that $x\leq y$, they can be seen as antisymmetric relations  compatible with the natural order on $\N$, \emph{i.e.} $(x,y)\in T$ implies $x\leq y$.\\
On the other hand some multitildes expressions are equivalent. For instance, $\widetilde{ab}\widetilde{cd}$ and $\widetilde{\widetilde{ab}\widetilde{cd}}$ denote the same language : $\{\epsilon, ab,cd,abcd\}$. This phenomenon is not very natural when stated in terms of relations. Nevertheless, up to a slight transformation, this can seen as a transitive closure.\\  
Let us be more precise and define a graduated bijection $\phi$ between $\mathcal{T}_k$ and $\mathrm{RAS}_{k}^\leq:=\mathrm{RAS}^\leq\cap \{1,\dots,k+1\}^2$ for any integer $k$:
  
%

  \centerline{
    $\phi:\left\{
	  \begin{array}{l@{\ }l@{\ }l}
	    \mathcal{T}_k & \longrightarrow & \mathrm{RAS}_{k}^\leq\\
	    T & \longrightarrow & \{(x,y+1)\mid (x,y)\in T\} \cup \{(x,x)\mid x\in\Intervalle{1}{n}\}
	  \end{array}
    \right.$
  }
\noindent The fact that $\phi$ is a bijection is obvious and the inverse bijection $\phi^{-1}$ is given by:\\
  \centerline{
    $\phi^{-1}:\left\{
	  \begin{array}{l@{\ }l@{\ }l}
	    \mathrm{RAS}_{k}^\leq & \longrightarrow & \mathcal{T}_k\\
	    R & \longrightarrow & \{(x,y)\mid (x,y+1)\in R\wedge x\neq y\}
	  \end{array}
    \right.$
  }
  
    
  
 
We endow the set ${\rm RAS}^\leq$ with a structure of operad, setting:
\[
R_1\circ'_i R_2=\phi(\phi^{-1}(R_1)\circ_i\phi^{-1}(R_2)).
\]

  Moreover, the bijection $\phi$ is in fact an isomorphism from the operad $(\mathcal{T},\circ)$ to $(\mathrm{RAS}^\leq,\circ')$.
  
We define the operator $\ShiftDiamnk{n,k}$ acting on pairs of integers $(x,y)$ by
\[
\ShiftDiam{n,k}{x,y}=\left\{
\begin{array}{ll}
(x,y)&\mbox{ if } y\leq k\\
(x,y+n-1)&\mbox{ if } x\leq k<y\\
(x+n-1,y+n-1)&\mbox{ otherwise}.
\end{array}
\right.
\]
We extend the definition to the set of pairs by $\ShiftDiam{n,k}S=\{\ShiftDiam{n,k}{x,y}:(x,y)\in S\}$.
Note $\ShiftDiamnk{n,k}$ acts by relabelling $x\rightarrow x+n-1$ when $x>k$. 
Hence, we set as for the multitildes:
\[
R_1\Diamond_k R_2:=\ShiftDiam{n,k}{R_1}\cup \Dec {k-1}{R_2},
\]
for $R_1\in \mathrm{RAS}^{\leq}_m$, $R_2\in \mathrm{RAS}^{\leq}_n$ and $k\leq m$.\\
The following lemma shows that $\mathrm{RAS}^\leq$ is closed under the action of the compositions $\Diamond_k$ and that these binary operators are compatible with the graduation.
\begin{lemma}
Let $R_1\in \mathrm{RAS}^\leq_m$ and  $R_2\in \mathrm{RAS}^\leq_n$ then $R_1\Diamond_k R_2\in \mathrm{RAS}^\leq_{m+n-1}$.
\end{lemma}
\begin{proof}
The antisymmetry is straightforward from the definition of $\Diamond_k$. So it suffices to prove the reflexivity, that is to check $(x,x)\in R_1\Diamond_k R_2$ for each $1\leq x\leq m+n$. But, from the definitions, $(x,x)\in \ShiftDiam{n,k}{R_1}$ when $1\leq x\leq k$ or $n+k\leq x\leq n+m$ and $(x,x)\in \Dec{k-1}{R_2}$ is $k\leq x\leq n+k$. So for each $1\leq x\leq m+n$ we obtain $(x,x)\in  R_1\Diamond_k R_2$. Hence, $R_1\Diamond_k R_2\in \mathrm{RAS}^\leq_{m+n-1}$.
\end{proof}

From the partial composition $\Diamond_k$ we define 
\[
\Diamond:\left\{ \begin{array}{ccc} \mathrm{RAS}^\leq_m\otimes\mathrm{RAS}^\leq_{n_1}\otimes\cdots\otimes
\mathrm{RAS}^\leq_{n_m}&\longrightarrow& \mathrm{RAS}^\leq_{n_1+\cdots+n_m}\\
R\Diamond(R_1,\dots,R_m)&=& (\dots((R\Diamond_m R_1)\Diamond_{m-1} R_{m-1})\dots)\Diamond_{1} R_1. \end{array}\right.
\]
And more precisely:
  \begin{proposition}
   We have $\circ'_k=\Diamond_k$ and $\circ'=\Diamond$.
  \end{proposition}
  \begin{proof}
   We have $\phi \Shiftnk{n,k}=\ShiftDiamnk{n,k}\phi$ and $\phi\Decnk{k}=\Decnk{k}\phi$. So $$\phi(T_1\circ_k T_2)=\phi(T_1)\Diamond_k\phi(T_2),$$
for each $T_1\in\mathcal{T}_m$, $T_2\in\mathcal{T}_n$ and $1\leq k\leq m$.
Since $\phi: (\mathcal{T},\circ)\rightarrow (\mathrm{RAS}^\leq,\circ')$ is an isomorphism of operads, we prove the result.
  \end{proof}  
It follows that:
\begin{theorem}
The structure $(\mathrm{RAS}^\leq,\Diamond)$ defines an operad isomorphic to $(\mathcal{T},\circ)$.
\end{theorem}

\subsection{From $\mathrm{RAS}^\leq$ to $\mathrm{POSet}^\leq$}

  Any reflexive and antisymmetric relation $R$ can be turned into a partial order when applying the transitive closure denoted by $\gamma(R)$.\\
Note, since $\ShiftDiamnk{n,k}$ and $\Decnk{k}$ are just relabelling, they commute with the operator $\gamma$:
\begin{equation}\label{CommGamma}
\gamma\ShiftDiamnk{n,k}=\ShiftDiamnk{n,k}\gamma,\, \Decnk{k}\gamma=\gamma\Decnk{k}.
\end{equation}

  

We need the following lemma
\begin{lemma}
\[
\gamma(\gamma R_1\Diamond_k \gamma R_2)=\gamma(R_1\Diamond_k R_2)
\]
\end{lemma}
\begin{proof}
It suffices to prove that $\gamma R_1\Diamond_k \gamma R_2\subset \gamma(R_1\Diamond_k R_2)$. Indeed by transitivity, this implies $\gamma(\gamma R_1\Diamond_k \gamma R_2)\subset \gamma(R_1\Diamond_k R_2)$ and the reverse inclusion is obvious.\\
We have $\gamma R_1\Diamond_k \gamma R_2=\ShiftDiam{k,n}{\gamma R_1}\cup \Dec{k-1}{\gamma R_2}$. From eq (\ref{CommGamma}) we obtain $\gamma R_1\Diamond_k \gamma R_2=(\gamma \ShiftDiam{k,n}{R_1})\cup (\gamma\Dec{k-1}{R_2})\subset \gamma(R_1\Diamond_k R_2)$, as expected.

\end{proof}

As an immediate consequence,  the transitive closure is compatible with the structure of operads that is:
\begin{lemma}
 If $\gamma(R_1)=\gamma(R'_1)$ and $\gamma(R_2)=\gamma(R'_2)$ then  $\gamma(R_1\Diamond_k R_2)=\gamma(R'_1\Diamond_k R'_2)$.
\end{lemma}

Denote by $\equiv_\gamma$ the equivalence relation on $\mathrm{RAS}^\leq$ defined by $R_1\equiv_\gamma R_2$ if and only if $\gamma(R_1)=\gamma(R_2)$. The quotient $\mathrm{RAS}^\leq/_{\equiv_\gamma}$ is automatically endowed with a structure of operad whose composition is deduced from  $\gamma$; we  denote by $\Diamond'$ the associated composition.\\
The equivalence classes are indexed by the element $\mathrm{POSet}^\leq:=\gamma(\mathrm{RAS}^\leq)$. Let us set $\Diamonddot_k:=\gamma\Diamond_k$ and $\Diamonddot:=\gamma\Diamond$.
 
  Immediately we obtain:
  \begin{proposition}
    The structure $(\mathrm{POSet}^\leq,\Diamonddot)$ is an operad isomorphic to $(\mathrm{RAS}^\leq/_{\equiv_\gamma},\Diamond')$.
  \end{proposition}
  
\subsection{Pseudotransitive tildes}
A multitilde is said \emph{pseudotransitive} if it is the image of a POSet by $\phi^{-1}$. The set of pseudotransitive multitildes will be denoted by
\[
{\rm PTT}:=\phi^{-1}({\rm POSet}^{\leq}).
\]
This set is endowed with an operad structure induced by the partial products
\[
T_1{\odot_k} T_2:=\phi^{-1}(\phi(T_1)\Diamonddot_k \phi(T_2))=\phi^{-1}\gamma\phi(T_1\circ_k T_2).
\]

For simplicity, set ${\tilde\gamma}:=\phi^{-1}\gamma\phi$. The equivalence relation $\equiv$ on $\mathcal T$ defined by $T_1\equiv T_2$ if and only if  $\tilde\gamma(T_1)=\tilde\gamma(T_2)$ is compatible with the operad structure. Indeed,

\[
\begin{array}{rcl}
\tilde\gamma(T_1\circ_k T_2)&=& \phi^{-1}\gamma\phi(T_1\circ_k T_2)\\
&=&  \phi^{-1}\gamma(\phi(T_1)\Diamond_k \phi(T_2))
\end{array}
\]
But, $\tilde\gamma(T)=\tilde\gamma(T')$ implies $\phi(T)\equiv_\gamma \phi(T')$. Since  $\equiv_\gamma$ is compatible with the composition in $\mathrm RAS^\leq$ we obtain
\[
\begin{array}{rcl}
\tilde\gamma(T_1\circ_k T_2)&=& \phi^{-1}\gamma(\phi(T'_1)\Diamond_k \phi(T'_2))\\
&=& \tilde\gamma(T'_1\circ_k T'_2).
\end{array}
\]
Hence, the quotient ${\mathcal T}/_\equiv$ is isomorphic to $({\rm PTT},\odot)$.

We explicitly describe  $\tilde\gamma$.
\begin{lemma}\label{pseudotrans}
The set $\tilde\gamma(T)$ is the smallest set $S$ such that the two following assertions are satisfied:
\begin{enumerate}
\item $T\subset S$ 
\item $(i,k), (k+1,j)\in S$ implies $(i,j)\in S$.
\end{enumerate}   
\end{lemma}
\begin{proof}
Since $\gamma(\phi(T))$ is the transitive closure of $\phi(T)$, it is the smallest set $R$ such that
\begin{enumerate}
\item $\phi(T)\in R$
\item $(i,k), (k,j)\in R$ implies $(i,j)\in R$.
\end{enumerate}
Hence, $\phi^{-1}(\gamma(\phi(T)))$ is such that if $(i,k), (k+1,j)\in \phi^{-1}(\gamma(\phi(T)))$ implies $(i,j)\in \phi^{-1}(\gamma(\phi(T)))$.\\
Indeed let $(i,k),\, (k+1,j)\in\tilde\gamma(T)$ then $(i,k+1),\ (k+1,j+1)\in\gamma(\phi(T))$. Since $\gamma(\phi(T))$ is transitive, we also have $(i,j+1)\in\gamma(\phi(T))$. The image of this pair is $\phi^{-1}(i,j+1)=(i,j)\in\tilde\gamma(T)$.\\
  Let $S$ be a subset of $\phi^{-1}(\gamma(\phi(T)))$ verifying
\begin{enumerate}
\item $T\subset S$ 
\item $(i,k), (k+1,j)\in S$ implies $(i,j)\in S$.
\end{enumerate} 
Suppose that $(i,j)\not\in \phi^{-1}(\gamma(\phi(T)))$, since $\phi$ is a bijection we obtain $\phi(i,j)\not\in \gamma(\phi(T))$. This contradicts the minimality of $\gamma(\phi(T))$ and prove the result.
\end{proof}

The set $\tilde\gamma(T)$ will be referred as the pseudotransitive closure of $T$.

\begin{lemma}\label{Vgamma2V}
For each tilde $T\in\mathcal T$,  $V(T)=V(\tilde\gamma(T))$ holds.
\end{lemma}
\begin{proof}
Suppose that $(i,k),\, (k+1,j)\in T$. Let us prove  $V(T\cup\{(i,j)\})=V(T)$. Suppose $(i,j)\not\in T$ (otherwise the result is obvious). We compare the set $E:=\{v(S): S\in \mathcal F(T)\}$ and $F:=\{v(S):S\in\mathcal F(T\cup\{(i,j)\})\}$. Obviously, $\mathcal F(T)\subset \mathcal F(T\cup\{(i,j)\})$ hence $E\subset F$. Let us prove the reverse inclusion and let $S\in \mathcal F(T\cup\{(i,j)\})$. If $(i,j)\not\in S$ then $S\in\mathcal F(T)$ and then $v(S)\in E$. Now suppose $(i,j)\in S$ and set $S'=S\cup\{(i,k),(k+1,j)\}\setminus \{(i,j)\}$. Obviously $S'\in\mathcal F(T)$.  Set $v(S)=(s_1,\dots,s_n)$ and $v(S')=(s'_1,\dots,s'_n)$. We have $s_l=s'_l$ if $l\not\in \llbracket i,j\rrbracket$ and $0$ otherwise. But $s'_l=0$ when $l\in\llbracket i,k\rrbracket\cup \llbracket k+1,j\rrbracket=\llbracket i,j\rrbracket$. So $v(S)=v(S')\in E$. 
\end{proof}
\noindent As a consequence, we define an action of ${\rm POSet}^{\leq}$ and $\rm PTT$ on languages:

\begin{theorem}
The sets $2^{\Sigma^*}$ and ${\mathrm Reg}(\Sigma^*)$ are ${\rm POSet}^{\leq}$-module and $\rm PTT$-module.
\end{theorem}
\begin{proof}
From Corollary \ref{cormod}, $2^{\Sigma^*}$ and ${\mathrm Reg}(\Sigma^*)$ are $\mathcal T$-module.\\
Note if $T_1\equiv T_2$ are two tildes of $\mathcal T$ then, by Lemma \ref{Vgamma2V}, $T_1$ and $T_2$ have the same action on  $2^{\Sigma^*}$ and ${\mathrm Reg}(\Sigma^*)$. These set are $\rm PTT$-modules. Since $\rm PTT$ is isomorphic to ${\rm POSet}^{\leq}$ as an operad, they are also ${\rm POSet}^{\leq}$-modules.
\end{proof}

  
  
  
  

 
\section{Consequences and perspectives\label{SecCoPer}}
 \subsection{Enumeration}

The purpose of this section is to compute an upper bound of the number $n(L_1,\dots,L_k)$ of languages  that can be obtained by the action of tilde on the $k$-tuple of languages $(L_1,\dots,L_k)$.\\
 First, Lemma \ref{Vgamma2V} implies that $n(L_1,\dots,L_k)$ equals to the number of languages obtained by applying a $\mathrm PTT$ on $(L_1,\dots,L_k)$.\\
From the previous section, the number of $\mathrm PTT$ with arity $k$ equals the numbers $p_k$ of POSets on $\{1,\dots,k+1\}$ that are contained in the usual linear order on integers. So, 
\begin{equation}
n(L_1,\dots,L_k)\leq p_k.
\end{equation}
The sequence of $p_k$ has no known closed form but the first values can be found on \cite{Sloane} (A006455); the first values are
 {
$$\begin{array}{l} 2, 7, 40, 357, 4824, 96428, 2800472, 116473461, 6855780268,\\ 565505147444, 64824245807684,\dots \end{array}$$
}
Note this is also the number of $(k+1)\times(k+1)$ upper triangular idempotent boolean matrices with all diagonal entries $1$.\\
 Let us show that the bound is reached for $L_1=\{a_1\},\dots, L_k=\{a_k\}$ where $a_1,\dots, a_k$ are distinct letters.\\
Let $P_1\neq P_2$ be two $\mathrm PTT$s with arity $k$ and $(i,j)\in P_1\setminus P_2$. We have $\{(i,j)\}\in {\mathcal F}(P_1)$ and then $(1^{i-1},0^{j-i},1^{n-j+1})\in V(P_1)$.\\
Suppose $(1^{i-1},0^{j-i},1^{n-j+1})\in V(P_2)$. This implies that there exists a set $\{(i,k_1),(k_1+1,k_2),\dots,(k_m+1,j)\}\in {\mathcal F}(P_2)$ for some $m$.  So $\{(i,k_1),(k_1+1,k_2),\dots,(k_m+1,j)\}\subset P_2$ and from Lemma \ref{pseudotrans} we obtain $(i,j)\in P_2$. This contradicts our hypothesis. Then, 
\begin{corollary}
Let $a_1,\dots, a_k$ be $k$ distinct letters. We have
\[
n(\{a_1\},\dots,\{a_k\})=p_k.
\]
\end{corollary}
\begin{example}
\rm There are $7$ languages obtained by applying $\mathrm PTT$ on a pair of letters $(\{a\},\{b\})$:
\[
\begin{array}{cc}
\mathrm tildes&languages\\\hline
\emptyset&ab\\
(1,1)&ab,\,b\\
(1,1), (1,2)&ab,\, b,\, \epsilon\\
(2,2)&ab,\,a\\
(2,2),\, (1,2)&ab,\ a,\ \epsilon\\
(1,1),\,(2,2),\,(1,2)&ab,\ a,\ b,\ \epsilon\\
(1,2)&ab,\, \epsilon
\end{array}
\]
\end{example}
\subsection{Finite languages}
Every finite languages can be generated by a multitilde acting on a sequences of elements in $\Sigma_{0}:=\{\{a\}\}_{a\in \Sigma}\cup \{\emptyset\}$.
\begin{proposition}
For any finite language $L$ over an alphabet $\Sigma$, it exists a multitilde $T\in\mathcal T_k$ and a $k$-tuples $(L_1,\dots,L_k)$ for some integer $k$, with $L_1,\dots, L_k\in\Sigma_{0}$ such that
\[
L=T(L_1,\dots,L_k).
\]
\end{proposition}
\begin{proof}
Note first, the set of finite languages is a $\mathcal T$-module, since it is contained in $2^{\Sigma^*}$, which by Corollary \ref{cormod} is a $\mathcal T$-module, and is closed by the actions of $\mathcal T$. 
Straightforwardly, we have $\emptyset=\emptyset_1(\emptyset)$, $\{\varepsilon\}=\{(1,1)\}(\emptyset)$ and $\{a\}=\emptyset_1(\{a\})$ for any $a\in \Sigma$.\\
 Hence, since $\mathcal T$ is an operad (Theorem \ref{TildeOperad}), it suffices to prove that the union and the catenation of two languages can be obtained by using a multitilde. We verify easily that
\[
L_1L_2=\emptyset_{2}(L_1,L_2)\mbox{ and } L_1\cup L_2=\{(1,2),(2,3)\}(L_1,\emptyset,L_2).
\]
This ends the proof.
\end{proof}

Let us give some examples.
\begin{example}\rm
The set of all subwords of a given word $a_1\cdots a_k$, is easily expressed as a multitilde by
\[
\left\{(i,i):i\in \llbracket 1,k\rrbracket\right\}(a_1,\dots,a_k).
\]
Note  the corresponding POSet is the natural order relation on $\llbracket 1,k+1\rrbracket$.
\end{example}
\begin{example}
\rm The set of all the prefixes (resp. suffices) of a given word $a_1\dots a_k$ admits a nice expression in terms of multitildes:
\[
P_k:=\left\{(i,k):i \in \llbracket 1,k\rrbracket\right\}(a_1,\dots,a_k) \, \left(\mbox{ resp.} 
S_k:=\left\{(1,i):i \in \llbracket 1,k\rrbracket\right\}(a_1,\dots,a_k)\right).
\]
The corresponding POSets are\\
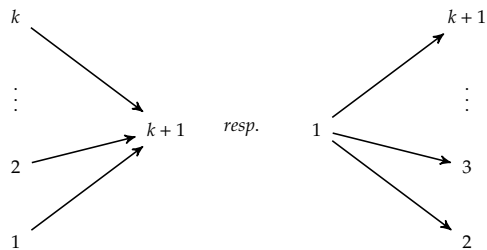
\begin{center}
\begin{tikzpicture}
\tikzstyle{VertexStyle}=[
]

\SetUpEdge[lw = 1.5pt,
 style={post},
labelstyle={sloped}
]
\tikzset{EdgeStyle/.style={post}}

\Vertex[x=0, y=0, 
 L={$1$}
]{p1}
\Vertex[x=0, y=1, 
 L={$2$}
]{p2}
\Vertex[x=0, y=2, 
 L={$\vdots$}
]{dd2}
\Vertex[x=0, y=3, 
 L={$k$}
]{pk}
\Vertex[x=2, y=1.5, 
 L={$k+1$}
]{pk+1}

\Edge[color=black](p1)(pk+1)
\Edge[color=black](p2)(pk+1)
\Edge[color=black](pk)(pk+1)

\Vertex[x=3, y=1.5, 
 L={$resp.$}
]{resp}

\Vertex[x=4, y=1.5, 
 L={$1$}
]{q1}
\Vertex[x=6, y=0, 
 L={$2$}
]{q2}
\Vertex[x=6, y=1, 
 L={$3$}
]{q3}
\Vertex[x=6, y=2, 
 L={$\vdots$}
]{dd3}
\Vertex[x=6, y=3, 
 L={$k+1$}
]{qk+1}
\Edge[color=black](q1)(qk+1)
\Edge[color=black](q1)(q2)
\Edge[color=black](q1)(q3)

\end{tikzpicture}
\end{center}
\end{example}
\begin{example}\rm
As a consequence, the set of all the factors of the word $a_1\dots a_k$ is obtained by the action of $F_k:=P_k\cup S_k$. Graphically, the corresponding POSet is:\\

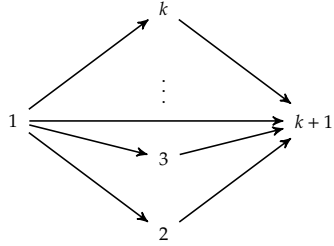
\begin{center}
\begin{tikzpicture}
\tikzstyle{VertexStyle}=[
]

\SetUpEdge[lw = 1.5pt,
 style={post},
labelstyle={sloped}
]
\tikzset{EdgeStyle/.style={post}}

\Vertex[x=0, y=1.5, 
 L={$1$}
]{p1}

\Vertex[x=2, y=0, 
 L={$2$}
]{p2}
\Vertex[x=2, y=1, 
 L={$3$}
]{p3}
\Vertex[x=2, y=2, 
 L={$\vdots$}
]{dd2}
\Vertex[x=2, y=3, 
 L={$k$}
]{pk}
\Vertex[x=4, y=1.5, 
 L={$k+1$}
]{pk+1}

\Edge[color=black](p1)(p2)
\Edge[color=black](p1)(p3)
\Edge[color=black](p1)(pk)
\Edge[color=black](p1)(pk+1)
\Edge[color=black](p2)(pk+1)
\Edge[color=black](p3)(pk+1)
\Edge[color=black](pk)(pk+1)
\end{tikzpicture}
\end{center}
For instance, consider the multitilde 
\[F_3:=\{(1,1),(1,2),(1,3),(2,3),(3,3)\}\]
its free subsets are
\[\{\emptyset\}\cup \{\{c\}:c\in F_3\}\cup \{\{(1,1),(2,3)\},\{(1,2),(3,3)\},\{(1,1),(3,3)\}\}.\]
The corresponding boolean vectors are
\[
(1,1,1), (0,1,1), (0,0,1), (0,0,0), (1,0,0), (1,1,0), (0,1,0).
\]
When acting on $(\{a\},\{b\},\{c\})$ they respectively give  $\{abc\}$, $\{bc\}$, $\{c\}$, $\{\varepsilon\}$, $\{a\}$, $\{ab\}$ and $\{b\}$.
\end{example}
\noindent
\begin{example}\rm
More generally, we have the following property:
\\
If $L=T(L_1,\dots,L_k)$ with $L_i\in \{\{a\}:a\in\Sigma\}\cup\{\emptyset\}$ we have
\begin{enumerate}
 \item The set of the prefixes of any words in $L$, ${\rm Pref}(L)$ is a subset of $(T\cup P_k)(L_1,\dots,L_k)$
  \item The set of the suffices of any words in $L$, ${\rm Suff}(L)$ is a subset of $(T\cup S_k)(L_1,\dots,L_k)$.
 \item The set of the factors of any words in $L$, ${\rm Fact}(L)$ is a subset of $ (T\cup F_k)(L_1,\dots,L_k)$.
\end{enumerate}
One has only to show the first assertion. Indeed, the second point  is obtained by symmetry from the first and the third is simply the composition of the first and the second.\\
The free subsets of $T\cup P_k$ are
\[
\mathcal F(T\cup P_k)=\mathcal F(T)\cup \{S'=S\cup \{(i,k)\}:i\in\llbracket 1,k\rrbracket, S\in\mathcal F(T),\ S' \mbox{ is free}\}.
\]
Consider a proper prefix $p$ of word $w=ps\in L$. Let $I\subset \llbracket 1,k\rrbracket$ such that $i\in I$ if and only if $L_i\neq \emptyset$. Without loss of generality, we suppose that all the $L_i$ for $i\in I$ are distinct.  Set $s=a_is'$ with $a_i\in\Sigma$ and let $j$ be the unique integer such that $L_j=\{a_j\}$. Since $w\in L$, it exists a free list $S$ such that $v(S)(L_1,\dots,L_k)=\{w\}$. Furthermore, $S'=\{(j_1,j_2): (j_1,j_2)\in S, \llbracket j_1,j_2\rrbracket\subset \llbracket 1,j\rrbracket\}\subset \mathcal F(T)$. Indeed, there is no pair $(j_1,j_2)$ in $S$ such that $j\in \llbracket j_1,j_2\rrbracket$ (otherwise the letter  $a_j$ does not appear in $w$). Hence, $S$ is composed only by pairs $(j_1,j_2)\in T$ such that $j_2<j$ or $j<j_1$. It follows that $S'=\{(j_1,j_2):j_2<j\}$ is a free list of $T$ because $S$ is free. So $S'\cup \{(j+1,k)\}\subseteq T\cup P_k$ is free and $v(S'\cup \{(j+1,k)\})(L_1,\dots,L_k)=v(S')(L_1,\dots,L_j)=p$. It follows that $p\in (T\cup P)(L_1,\dots,L_k)$ and our claim.\\ 
The inclusions are strict in the general case. For instance, consider $\emptyset_4(\{a\},\{b\},\{c\},\emptyset)=\emptyset$ and $(\emptyset_4\cup P_4)(\{a\},\{b\},\{c\},\emptyset)= P_4(\{a\},\{b\},\{c\},\emptyset)=\{\varepsilon,a,ab,abc\}$.\\
Note if each $L_i\neq\emptyset$ the inclusions are equalities. Although this is a sufficient condition, it is not necessary as shown by $\{(1,3),(3,5)\}(\{a\},\{b\},\emptyset,\{c\},\{d\})=\{ab,cd\}$ and
\[
(\{(1,3),(3,5)\}\cup P_5)(\{a\},\{b\},\emptyset,\{c\},\{d\})=
\{\varepsilon,\{a\}, \{ab\}, \{c\}, \{cd\} \}={\mathrm{Pref}}(\{ab,cd\}).
\]
\end{example}


\subsection{Regular languages}
In the previous section, we showed that each finite language can be obtained by the action of a multitilde on a $k$-tuple in $\Sigma_0^k$. As a consequence, any regular languages can be written as a combination of multitildes and stars acting on a $k$-tuple in $\Sigma_0^k$.\\
Consider the smallest operad $\mathcal T^\star$ containing $\mathcal T$ together a new formal $1$-ary operator denoted $\star$. The elements of $\mathcal T^\star$ are trees whose nodes are multitildes or $\star$ and are such that the roots of the subtrees immediately issued of a multitilde are not  multitildes (indeed, the composition $\circ_i$ of two multitildes gives an other multitilde).\\ To define properly an action of $\mathcal T^\star$ on languages, it suffices to define the action  of the operator $\star$ which is naturally $\star(L)=L^*$. Note $\star(\star(L))=(L^*)^*=L^*=\star(L)$. To simplify the operators we introduce the smallest equivalence relation $\equiv$ compatible with the composition such that $\star \circ \star\equiv \star$, and the operad $T^\star=\mathcal T^\star/_\equiv$. Straightforwardly, the following result holds:
\begin{proposition}
\begin{enumerate}
 \item  The sets $2^{\Sigma^*}$ and $Reg(\Sigma^*)$ are both $T^\star$ and $\mathcal T^\star$-modules.
 \item Each regular language can be obtained by the action of an element in $T^\star$ on a $k$-tuple in $\Sigma_0^k$.
\end{enumerate}
\end{proposition}
Similarly to the operad $\mathcal T$, there is a very interesting underlying combinatoric structure  to the operad $T^\star$ involving pairs of relations and a generalization of the dissections of polygons. This study is differed to a forthcoming paper.

\section{Conclusion}

In this paper, we   introduced several operads which allow us to represent multitilde operations defined  by one of the authors in \cite{CCM11e}. Each of these operads plays a specific role in the understanding of these tools. For instance, the operad $\mathcal T$ encodes the structure of composition of the multitilde operators. The operad $\mathcal B$ explains how the multitildes act on languages. Finally, the operads $\mathrm RAS^\leq$ and $\mathrm POSet^\leq$ establish a relation between these objects and well known combinatorial structures. This formalizes the equivalence of two multitildes and enumerates the different inequivalent multitildes.

This study is a first step towards the construction of a new combinatorial approach of the theory of languages. First, regular languages can be seen as the result of the action of an element of $T^\star$ on a $k$-tuples of letters. Second, another action of $\mathcal T$ on languages have been defined in \cite{CCM11e} and combined with the action of multitildes in \cite{CCM10}. Similar operadic constructions can be defined for these operators.  These two points will be studied in forthcoming papers.

The multitilde have been introduced in the aim to  enlarge the class of automata which can be represented as an expression using a number of symbols which is a linear function of the number of states. It remains to study the algorithmic improvement which can be brought by the different structures of operads. 

\bibliography{biblio}

\end{document}